\documentclass{article}
\usepackage{amssymb}
\usepackage{graph}
\newtheorem{theorem}{Theorem}
\newtheorem{lemma}{Lemma}
\newtheorem{observation}{Observation}
\newtheorem{corollary}{Corollary}
\newtheorem{remark}{Remark}
\newtheorem{definition}{Definition}
\newenvironment{proof}{\par\noindent\textit{Proof.} }{\hfill\framebox(5,5){}\medskip\par}
\newenvironment{myenumeratei}{
\begin{enumerate}
\vspace{-.05in}
\itemsep -.025in

}
{\end{enumerate}}
\newenvironment{myenumeratea}{
\begin{enumerate}
\vspace{-.05in}
\itemsep -.025in

}
{\end{enumerate}}

\newcommand{\com}[1]{}
\begin{document}
\begin{center}
\begin{tabular}{c}
{\Large Uniquely Restricted Matchings in Interval Graphs}
\vspace{.2in}\\
{Mathew C. Francis\footnotemark[1]\hspace{.5in} Dalu Jacob\footnotemark[2]\hspace{.5in} Satyabrata Jana\footnotemark[3]}\vspace{.1in}\\
{Indian Statistical Institute, Chennai Centre}
\end{tabular}
\end{center}
\footnotetext[1]{E-mail: \texttt{mathew@isichennai.res.in}}
\footnotetext[2]{E-mail: \texttt{dalu1991@gmail.com}}
\footnotetext[3]{E-mail: \texttt{satyamtma@gmail.com}}
\vspace{.2in}

\bibliographystyle{plain}
\begin{abstract}
A matching $M$ in a graph $G$ is said to be uniquely restricted if there is no other matching in $G$ that matches the same set of vertices as $M$. We describe a polynomial-time algorithm to compute a maximum cardinality uniquely restricted matching in an interval graph, thereby answering a question of Golumbic et al. (``Uniquely restricted matchings'', M.~C.~Golumbic, T.~Hirst and M.~Lewenstein, \textit{Algorithmica}, 31:139--154, 2001). Our algorithm actually solves the more general problem of computing a maximum cardinality ``strong independent set'' in an interval nest digraph, which may be of independent interest. Further, we give linear-time algorithms for computing maximum cardinality uniquely restricted matchings in proper interval graphs and bipartite permutation graphs.
\end{abstract}
\section{Introduction}
Let $G = (V,E)$ be a graph. A set of edges $M \subseteq E(G)$ is said to be a \emph{matching} if no two edges of $M$ share a common vertex. The set of vertices in $V$ that have an edge of $M$ incident on them are called the \emph{vertices matched by $M$}. A matching $M$ is said to be \emph{uniquely restricted} if there is no other matching that matches the same set of vertices as $M$. The problem of finding a maximum cardinality uniquely restricted matching in an input graph is known to be NP-complete even for the special cases of split graphs and bipartite graphs~\cite{GHL01}. In their paper initiating the study of uniquely restricted matchings, Golumbic, Hirst and Lewenstein~\cite{GHL01} present linear-time algorithms for the problem on threshold graphs, proper interval graphs, cacti and block graphs, while leaving open the question of whether polynomial-time algorithms exist for the problem on interval graphs and permutation graphs. In Section~\ref{sec:int}, we answer the question for interval graphs by constructing a polynomial-time algorithm that computes a maximum cardinality uniquely restricted matching in any interval graph. The algorithm is actually a dynamic programming algorithm that computes a maximum cardinality ``strong independent set'' in an interval nest digraph (see Section~\ref{sec:int} for definitions)---this problem is shown to be more general than the problem of computing a maximum cardinality uniquely restricted matching in interval graphs. Before that, in Sections~\ref{sec:propint} and~\ref{sec:bipperm}, we present linear-time dynamic programming algorithms for computing a maximum cardinality uniquely restricted matching in proper interval graphs and bipartite permutation graphs respectively. We note that the linear-time algorithm described for the problem for proper interval graphs in~\cite{GHL01} does not appear to work in all cases.  

\section{Preliminaries}
We consider only finite graphs.
Wherever it is not specified otherwise, ``graph'' shall mean an undirected graph. An edge between vertices $u$ and $v$ in an undirected graph is denoted as $uv$ and a directed edge (arc) from $u$ to $v$ in a directed graph (digraph) is denoted as $(u,v)$. We shall denote the vertex set and edge set of a graph or digraph $G=(V,E)$ by $V(G)$ and $E(G)$ respectively. All graphs and digraphs considered are simple---i.e., there are no loops or multiple edges.

Let $M$ be a matching in a graph $G$. An even cycle in $G$ is said to be an \emph{alternating cycle with respect to $M$} if every second edge of the cycle belongs to $M$~\cite{GHL01}. The following theorem characterizes uniquely restricted matchings in terms of alternating cycles.

\begin{theorem}[\cite{GHL01}]\label{thm:urmaltcycle}
Let $G=(V,E)$ be a graph. A matching $M$ in $G$ is uniquely restricted if and only if there is no alternating cycle with respect to $M$ in $G$.
\end{theorem}

\begin{lemma}\label{lem:altc4}
Let $G=(V,E)$ be any graph and let $\{uv,u'v'\}$ be a matching in it. The following statements are equivalent:
\begin{myenumeratei}
\item\label{it:c4} There is a cycle of length 4 containing the edges $uv$ and $u'v'$ in $G$.
\item\label{it:altc4} There is an alternating cycle with respect to $\{uv,u'v'\}$ in $G$.
\item\label{it:neigh} Each of $u,v$ has at least one neighbour in $\{u',v'\}$ and each of $u',v'$ has at least one neighbour in $\{u,v\}$.
\end{myenumeratei}
\end{lemma}
\begin{proof}
\ref{it:c4}$\Leftrightarrow$\ref{it:altc4} is straightforward. We shall therefore only prove \ref{it:altc4}$\Leftrightarrow$\ref{it:neigh}.

If there is an alternating cycle with respect to the matching $\{uv,u'v'\}$, then it is either $uvu'v'u$ or $uvv'u'u$. In any case, it is clear that each of $u,v$ has at least one neighbour in $\{u',v'\}$ and each of $u',v'$ has at least one neighbour in $\{u,v\}$.

Now suppose that each of $u,v$ has at least one neighbour in $\{u',v'\}$ and each of $u',v'$ has at least one neighbour in $\{u,v\}$. Then, $(vu'\notin E(G)$ or $v'u\notin E(G))\Rightarrow (uu'\in E(G)$ and $vv'\in E(G))$. This means that if $vu'\notin E(G)$ or $v'u\notin E(G)$, then $uvv'u'u$ is an alternating cycle with respect to $\{uv,u'v'\}$ in $G$. If on the other hand both $vu',v'u\in E(G)$, then $uvu'v'u$ is an alternating cycle with respect to $\{uv,u'v'\}$ in $G$.
\end{proof}

\section{Proper Interval Graphs}\label{sec:propint}
By ``interval'' we shall mean a closed interval on the real line. An interval is denoted as $[a,b]$ where $a,b\in\mathbb{R}$ and $a\leq b$, and is the set $\{x\in\mathbb{R}\colon a\leq x\leq b\}$.
Given a graph $G$, a collection $\{I_u\}_{u\in V(G)}$ of intervals is said to be an \emph{interval representation} of $G$ if for distinct $u,v\in V(G)$, we have $uv\in E(G)$ if and only if $I_u\cap I_v\neq\emptyset$. Graphs which have interval representations are called \emph{interval graphs}. The following theorem about uniquely restricted matchings in interval graphs is from~\cite{GHL01}.

\begin{theorem}[\cite{GHL01}]\label{thm:urmintc4}
Let $G=(V,E)$ be an interval graph. Let $M$ be a matching $M$ in $G$. Then the following statements are equivalent:
\begin{myenumeratei}
\item\label{it:Murm} $M$ is uniquely restricted
\item\label{it:noaltC4} There is no alternating cycle of length 4 with respect to $M$ in $G$
\item\label{it:ee'urm} For any two edges $e,e'\in M$, $\{e,e'\}$ is a uniquely restricted matching in $G$
\end{myenumeratei}
\end{theorem}
\begin{proof}
We shall first show \ref{it:Murm}$\Leftrightarrow$\ref{it:noaltC4}. By Theorem~\ref{thm:urmaltcycle}, if $M$ is uniquely restricted, then there is no alternating cycle of any length with respect to $M$ in $G$. So it suffices to show that when $M$ is not uniquely restricted, there is an alternating cycle of length 4 with respect to $M$ in $G$. Let $<$ be an ordering of the vertices of $G$ according to a non-decreasing order of the left endpoints of their intervals in an interval representation of $G$. It is easy to see (and folklore) that the ordering $<$ has the property that for any $u,v,w\in V(G)$ such that $u<v<w$, $uw\in E(G)\Rightarrow uv\in E(G)$. Suppose that the matching $M$ is not uniquely restricted. Then, by Theorem~\ref{thm:urmaltcycle}, there exists some alternating cycle with respect to $M$ in $G$. Let $C=u_1u_2u_3\ldots u_ku_1$ be an alternating cycle with respect to $M$ of smallest possible length in $G$. Clearly, $k$ is even and $4\leq k\leq |V(G)|$. If $k=4$, then this cycle is an alternating cycle of length 4 with respect to $M$ and we are done. So let us suppose for the sake of contradiction that $k>4$. If $u_iu_j\in E(G)$ for some two vertices $u_i$ and $u_j$ that are not consecutive on the cycle $C$ (i.e., $u_iu_j$ is a ``chord'' of $C$) and $i$ and $j$ are of different parity, then one of the two cycles into which the chord $u_iu_j$ splits $C$ will be an alternating cycle with respect to $M$ of length smaller than $k$. As this is a contradiction to the assumption that $C$ is the alternating cycle with respect to $M$ of smallest possible length, we can assume that such chords are ``forbidden'', or in other words, there are no such chords for $C$. We shall also assume without loss of generality that $u_1=\max_<\{u_1,u_2,\ldots,u_k\}$ and that $u_k<u_2$ (otherwise we can relabel the vertices of the cycle $C$ to satisfy both these conditions). From the special property of the ordering $<$ and the fact that $u_ku_1\in E(G)$, it can be seen that if $u_k<u_3$, then $u_ku_3\in E(G)$. On the other hand, if $u_3<u_k$, we again have $u_ku_3\in E(G)$ as $u_3u_2\in E(G)$. As $u_ku_3$ is a forbidden chord, we have a contradiction.

The proof for \ref{it:noaltC4}$\Leftrightarrow$\ref{it:ee'urm} follows directly from Lemma~\ref{lem:altc4} and Theorem~\ref{thm:urmaltcycle}.
\end{proof}

If every interval in an interval representation is of unit length, then it is said to be a \emph{unit interval representation}. \emph{Unit interval graphs} are the graphs which have unit interval representations. Unit interval graphs are also called \emph{proper interval graphs} as these are exactly the graphs that have \emph{proper interval representations}---interval representations in which no interval strictly contains another interval.

\begin{definition}\label{def:proporder}
For a graph $G=(V,E)$, an ordering $<$ of $V(G)$ is said to be a \emph{proper vertex ordering} if for $u,v,w\in V(G)$  such that $u<v<w$, $uw\in E(G)\Rightarrow uv,vw\in E(G)$.
\end{definition}

Proper vertex orderings are called ``proper orderings'' in~\cite{GHL01}.
Note that in a proper interval representation of a graph $G$, the ordering of the vertices of $G$ according to the left endpoints of the intervals corresponding to them is the same as their ordering according to the right endpoints of their intervals, since no interval is contained in another. It is easy to see that this ordering is a proper vertex ordering of $G$. In fact, it is folklore that a graph is a unit interval graph if and only if it has a proper vertex ordering.

In~\cite{GHL01}, a linear-time algorithm to compute a maximum cardinality uniquely restricted matching in proper interval graphs is presented. But this algorithm fails to find out the exact solution in some cases (see Fig.~\ref{fig:counter}). In this section, we present a linear time algorithm for computing a maximum cardinality uniquely restricted matching in a given proper interval graph $G$.

\begin{figure}
\renewcommand{\vertexset}{(1,0,0),(2,1.5,0),(3,3,0),(4,4.5,0),(5,6,0),(6,7.5,0),(7,9,0)}
\renewcommand{\edgeset}{(1,2,,2),(2,3),(3,4),(4,5),(5,6),(6,7,,2),(1,3,,,.5),(1,4,,,1),(2,4,,,.5),(3,5,,2,1),(4,6,,,.5),(4,7,,,1),(5,7,,,.5)}
\renewcommand{\radius}{.1}
\begin{center}
\begin{tikzpicture}
\drawgraph
\node at (0,-.3) {1};
\node at (1.5,-.3) {2};
\node at (3,-.3) {3};
\node at (4.5,-.3) {4};
\node at (6,-.3) {5};
\node at (7.5,-.3) {6};
\node at (9,-.3) {7};
\end{tikzpicture}
\end{center}
\caption{A unit interval graph with vertices arranged according to a proper vertex ordering. The bold edges represent a maximum cardinality uniquely restricted matching. The algorithm from~\cite{GHL01}, given a proper vertex ordering as input, always produces a uniquely restricted matching consisting of edges between consecutive vertices. But every such matching in this graph, given this particular vertex ordering, has at most two edges.}\label{fig:counter}
\end{figure}
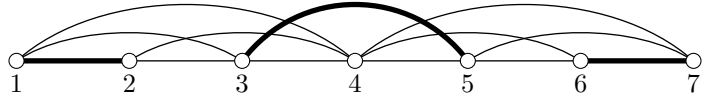

We can assume that $G$ is connected as if it is not, we can easily run the algorithm separately in each component to find a maximum cardinality uniquely restricted matching in each of them and then take a union of those to obtain a maximum cardinality uniquely restricted matching in $G$. We can also assume that a proper interval representation of the graph $G$, and therefore a proper vertex ordering of $G$, is available, as there are well known linear-time algorithms that can generate the proper interval representation of a graph, given its adjacency list~\cite{Corneilproper}.
\medskip

Let $<$ be a proper vertex ordering of a graph $G$. For an edge $e=uv\in E(G)$, we define $l_<(e)=\min_<\{u,v\}$ and $r_<(e)=\max_<\{u,v\}$. We shorten $l_<(e)$ and $r_<(e)$ to just $l(e)$ and $r(e)$ when the proper vertex ordering $<$ is clear from the context.

\begin{lemma}[\cite{GHL01}]\label{lem:unitintdisjurm}
Let $G$ be a proper interval graph with a proper vertex ordering $<$. If $\{e,e'\}$ is a uniquely restricted matching in $G$, then either $r(e)<l(e')$ or $r(e')<l(e)$.
\end{lemma}
\begin{proof}
As $\{e,e'\}$ is a matching, we can be sure that $l(e),r(e),l(e'),r(e')$ are all distinct vertices. Suppose that $l(e')<r(e)$ and $l(e)<r(e')$. Let us assume without loss of generality that $l(e)<l(e')$. Since $l(e)r(e)\in E(G)$ and $l(e)<l(e')<r(e)$, by Definition~\ref{def:proporder}, we have $l(e)l(e'),r(e)l(e')\in E(G)$. Now if $r(e)<r(e')$, then as $l(e')<r(e)<r(e')$ and $l(e')r(e')\in E(G)$, we have by Definition~\ref{def:proporder} that $l(e')r(e),r(e')r(e)\in E(G)$. On the other hand, if $r(e')<r(e)$, then as $l(e)<r(e')<r(e)$ and $l(e)r(e)\in E(G)$, Definition~\ref{def:proporder} gives us $l(e)r(e'),r(e)r(e')\in E(G)$. Thus, in any case, each of $l(e),r(e)$ has a neighbour in $\{l(e'),r(e')\}$ and each of $l(e'),r(e')$ has a neighbour in $\{l(e),r(e)\}$. We can now use Lemma~\ref{lem:altc4} to conclude that there is an alternating cycle with respect to $\{e,e'\}$ in $G$. But as $\{e,e'\}$ is a uniquely restricted matching, this contradicts Theorem~\ref{thm:urmaltcycle}.

\end{proof}

\begin{lemma}\label{lem:unitinturmpair}
Let $G$ be a proper interval graph with a proper vertex ordering $<$ and let $e,e'\in E(G)$. Then $\{e,e'\}$ is a uniquely restricted matching in $G$ if and only if $l(e),l(e')$ are distinct and nonadjacent or $r(e),r(e')$ are distinct and nonadjacent.
\end{lemma}
\begin{proof}
Suppose that $\{e,e'\}$ is a uniquely restricted matching. Then clearly $l(e),r(e),l(e'),r(e')$ are all distinct vertices. If $l(e)l(e')\in E(G)$ and $r(e)r(e')\in E(G)$, then we have the alternating cycle $l(e)l(e')r(e')r(e)$ $l(e)$ with respect to $\{e,e'\}$ in $G$, which is a contradiction to Theorem~\ref{thm:urmaltcycle}.

Let us now prove the other direction of the claim. Suppose that $l(e)$ and $l(e')$ are distinct and nonadjacent. We shall assume without loss of generality that $l(e)<l(e')$. As $l(e)l(e')\notin E(G)$ and $l(e)<l(e')<r(e')$, we have from Definition~\ref{def:proporder} that $l(e)$ has no neighbour in $\{l(e'),r(e')\}$. By Lemma~\ref{lem:altc4} and Theorem~\ref{thm:urmaltcycle}, this implies that $\{e,e'\}$ is a uniquely restricted matching in $G$. Now let us suppose that $r(e)$ and $r(e')$ are distinct and nonadjacent. Again, we shall assume without loss of generality that $r(e)<r(e')$. As $l(e)<r(e)<r(e')$ and $r(e)r(e')\notin E(G)$, we have from Definition~\ref{def:proporder} that $r(e')$ has no neighbour in $\{l(e),r(e)\}$. Lemma~\ref{lem:altc4} and Theorem~\ref{thm:urmaltcycle} can now be used to infer that $\{e,e'\}$ is a uniquely restricted matching in $G$.
\end{proof}

\begin{lemma}\label{lem:e1e2e3}
Let $G$ be a proper interval graph with a proper vertex ordering $<$. 
Let $e_1,e_2,e_3$ be distinct edges of $G$ such that $l(e_1)\leq l(e_2)\leq l(e_3)$ and $r(e_1)\leq r(e_2)\leq r(e_3)$. If $\{e_1,e_3\}$ is not a uniquely restricted matching in $G$, then neither $\{e_1,e_2\}$ nor $\{e_2,e_3\}$ is a uniquely restricted matching in $G$.
\end{lemma}
\begin{proof}
We shall show that $\{e_2,e_3\}$ is not a uniquely restricted matching in $G$. The proof for the case of $\{e_1,e_2\}$ is similar and is left to the reader. If $\{e_2,e_3\}$ is not even a matching in $G$, then we are immediately done. So we shall assume otherwise---i.e., $l(e_2),r(e_2),l(e_3),r(e_3)$ are all distinct vertices of $G$. This means that $l(e_1)\leq l(e_2)<l(e_3)$ and $r(e_1)\leq r(e_2)<r(e_3)$. By Lemma~\ref{lem:unitinturmpair}, we now have $l(e_1)l(e_3),r(e_1)r(e_3)\in E(G)$. This implies by Definition~\ref{def:proporder} that $l(e_2)l(e_3),r(e_2)r(e_3)\in E(G)$. Lemma~\ref{lem:unitinturmpair} now implies that $\{e_2,e_3\}$ is not a uniquely restricted matching.
\end{proof}

From Lemma~\ref{lem:unitintdisjurm}, it follows that the edges of any uniquely restricted matching $M$ in a proper interval graph $G$ with a proper vertex ordering $<$ can be labelled as $e_1,e_2,\ldots,e_{|M|}$ such that $l(e_1)<r(e_1)<l(e_2)<r(e_2)<\cdots<l(e_{|M|})<r(e_{|M|})$.
We say that the uniquely restricted matching $M$ \emph{starts with} the edge $e_1$.

We now give a stronger version of Theorem~\ref{thm:urmintc4} for the case of proper interval graphs.

\begin{theorem}\label{thm:consecedges}
Let $G$ be a proper interval graph and $<$ a proper vertex ordering of it.
Let $M=\{e_1, e_2,\ldots,e_t\}$ be a matching in $G$ where $l(e_1)<l(e_2)<\cdots<l(e_t)$. The matching $M$ is uniquely restricted in $G$ if and only if $\{e_i,e_{i+1}\}$ is a uniquely restricted matching in $G$, for each $i \in \{1,2,\ldots,t-1\}$.
\end{theorem}
\begin{proof}
As every subset of a uniquely restricted matching is also a uniquely restricted matching, to prove the theorem, it is sufficient to show that whenever $M$ is not a uniquely restricted matching, $\{e_i,e_{i+1}\}$ is not a uniquely restricted matching in $G$ for some $i\in\{1,2,\ldots,t-1\}$. Suppose that $M$ is not a uniquely restricted matching. By Theorem~\ref{thm:urmintc4}, there exists some subset $\{e_p,e_q\}$ of $M$, where $1\leq p<q\leq t$, such that $\{e_p,e_q\}$ is not a uniquely restricted matching.
We choose $p$ and $q$ such that for any $p',q'$ with $1\leq p'<q'\leq t$ and $q'-p'<q-p$, $\{e_{p'},e_{q'}\}$ is a uniquely restricted matching in $G$. Suppose that $q>p+1$. Observing that $l(e_p)<l(e_{q-1})<l(e_q)$, we can deduce that if $r(e_q)<r(e_{q-1})$ or $r(e_{q-1})<r(e_p)$, then by Lemma~\ref{lem:unitintdisjurm}, either $\{e_{q-1},e_q\}$ or $\{e_p,e_{q-1}\}$ respectively is not a uniquely restricted matching, in each case contradicting our choice of $p$ and $q$. Therefore, we get $r(e_p)<r(e_{q-1})<r(e_q)$. Now we can apply Lemma~\ref{lem:e1e2e3} to conclude that $\{e_{q-1},e_q\}$ is not a uniquely restricted matching, again contradicting our choice of $p$ and $q$. Thus we infer that $q=p+1$. For $i=p$, we now have that $\{e_i,e_{i+1}\}$ is not a uniquely restricted matching, thereby completing the proof.
\end{proof}

\begin{corollary}\label{cor:augment}
Let $G$ be a proper interval graph with a proper vertex ordering $<$ and let $M$ be a uniquely restricted matching in $G$ starting with an edge $e'\in E(G)$. Let $e\in E(G)$ be such that $r(e)<l(e')$ and $\{e,e'\}$ is a uniquely restricted matching in $G$. Then $\{e\}\cup M$ is a uniquely restricted matching in $G$ starting with $e$.
\end{corollary}
\begin{proof}
The proof follows directly from Theorem~\ref{thm:consecedges}.
\end{proof}
\medskip

From here onwards, we assume that $G$ is a connected proper interval graph with a proper vertex ordering $<$. Let $V(G)=\{v_1,v_2,\ldots,v_n\}$ where $v_1<v_2<\cdots<v_n$.

\begin{observation}\label{obs:consecedge}
For $1\leq i<n$, $v_iv_{i+1}\in E(G)$.
\end{observation}
\begin{proof}
Suppose that for some $i\in\{1,2,\ldots,n-1\}$, $v_iv_{i+1}\notin E(G)$. If for some pair of vertices $v_p,v_q$, where $1\leq p\leq i<i+1\leq q\leq n$, we have $v_pv_q\in E(G)$, then from Definition~\ref{def:proporder}, $v_pv_q\in E(G)\Rightarrow v_pv_{i+1}\in E(G)\Rightarrow v_iv_{i+1}\in E(G)$, which is a contradiction. Therefore, there does not exist any edge $v_pv_q\in E(G)$ such that $p\leq i$ and $q\geq i+1$. But this would mean that there is no edge in $G$ between a vertex in $\{v_1,v_2,\ldots,v_i\}$ and a vertex in $\{v_{i+1},v_{i+2},\ldots,v_n\}$, implying that $G$ is disconnected. This contradicts our assumption that $G$ is a connected graph.
\end{proof}
\medskip

For $u\in V(G)$, we define $\lambda(u)=\min_<\{v\in N(u)\cup\{u\}\}$ and $\rho(u)=\max_<\{v\in N(u)\cup\{u\}\}$.
We now associate a pair of edges with each edge of $G$. The following observation is an easy consequence of the property of proper vertex orderings given in Definition~\ref{def:proporder}.

\begin{observation}\label{obs:lambdarho}
For any vertex $u\in V(G)$, $N(u)=\{x\in V(G)\colon \lambda(u)\leq x\leq\rho(u)\}$.
\end{observation}

Let $e\in E(G)$. The \emph{left successor} of $e$, denoted by $\sigma_l(e)$, is defined to be the edge $v_{i+1}v_{i+2}$, where $v_i=\rho(l(e))$. It is clear from Observation~\ref{obs:consecedge} that $\sigma_l(e)$ exists in $G$ if and only if $\rho(l(e))<v_{n-1}$.
The \emph{right successor} of $e$, denoted by $\sigma_r(e)$, is defined to be the edge $\lambda(v_{i+1})v_{i+1}$ where $v_i=\rho(r(e))$. Note that for any vertex $u\in V(G)$ for which $\lambda(u)\neq u$, we have $\lambda(u)u\in E(G)$. Moreover, by Observation~\ref{obs:consecedge}, it follows that for every vertex $u\in V(G)\setminus\{v_1\}$, $\lambda(u)\neq u$. Therefore, it can be concluded that $\sigma_r(e)$ exists in $G$ if and only if $\rho(r(e))<v_n$.

\begin{observation}\label{obs:esuccurm}
Let $e\in E(G)$.
\begin{myenumeratea}
\item\label{it:elsurm} If $\sigma_l(e)$ exists, then $r(e)<l(\sigma_l(e))$ and $\{e,\sigma_l(e)\}$ is a uniquely restricted matching in $G$.
\item\label{it:ersurm} If $\sigma_r(e)$ exists, then $r(e)<l(\sigma_r(e))$ and $\{e,\sigma_r(e)\}$ is a uniquely restricted matching in $G$.
\end{myenumeratea}
\end{observation}
\begin{proof}
We shall first prove~\ref{it:elsurm}. As $l(e)r(e)\in E(G)$, we know that $\rho(l(e))\geq r(e)$. By definition of $\sigma_l(e)$, we have $r(e)\leq\rho(l(e))<l(\sigma_l(e))$. Therefore, we have $l(e)l(\sigma_l(e))\notin E(G)$. This implies by Lemma~\ref{lem:unitinturmpair} that $\{e,\sigma_l(e)\}$ is a uniquely restricted matching in $G$.

Now let us prove~\ref{it:ersurm}. It is clear from the definition of $\sigma_r(e)$ that $\rho(r(e))<r(\sigma_r(e))$. Therefore, we have $r(e)<r(\sigma_r(e))$ and $r(e)r(\sigma_r(e))\notin E(G)$. From Lemma~\ref{lem:unitinturmpair}, we can now conclude that $\{e,\sigma_r(e)\}$ is a uniquely restricted matching in $G$. By Lemma~\ref{lem:unitintdisjurm}, we also get that $r(e)<l(\sigma_r(e))$.
\end{proof}

Now for each edge $e\in E(G)$ we define a set $U(e)$ of edges as follows.
$$U(e)=\left\{\begin{array}{cl}\{e\}&\mbox{if neither }\sigma_l(e)\mbox{ nor }\sigma_r(e)\mbox{ exist}\\
\{e\}\cup U(\sigma_l(e))&\mbox{if }\sigma_r(e)\mbox{ does not exist, or if both exist and }|U(\sigma_l(e))|\geq |U(\sigma_r(e))|\\
\{e\}\cup U(\sigma_r(e))&\mbox{if }\sigma_l(e)\mbox{ does not exist, or if both exist and }|U(\sigma_l(e))|<|U(\sigma_r(e))|
\end{array}\right.$$

From Observation~\ref{obs:esuccurm}, we know that $r(e)<l(\sigma_l(e))$ and $r(e)<l(\sigma_r(e))$. This means that $U(e)$ is well-defined. The next two lemmas will show that $U(e)$ is always a uniquely restricted matching starting with $e$ of maximum possible cardinality (among the uniquely restricted matchings starting with $e$ in $G$).

\begin{lemma}\label{lem:Uisurm}
For any edge $e\in E(G)$, $U(e)$ is a uniquely restricted matching starting with $e$.
\end{lemma}
\begin{proof}
We shall prove this by induction on $|U(e)|$. It is clear from the definition of $U(e)$ that $|U(e)|\geq 1$.
If $|U(e)|=1$, then it must be the case that $U(e)=\{e\}$. The statement of the lemma is easily seen to be true in this case.
We shall now assume that $|U(e)|>1$ and that the statement of the lemma has been shown to be true for all $e'$ such that $|U(e')|<|U(e)|$. In this case, from the definition of $U(e)$, one of the following occurs: either (a) $\sigma_l(e)$ exists and $U(e)=\{e\}\cup U(\sigma_l(e))$, or (b) $\sigma_r(e)$ exists and $U(e)=\{e\}\cup U(\sigma_r(e))$.
If (a) occurs, then we have $|U(\sigma_l(e))|=|U(e)|-1$, and therefore by the induction hypothesis, $U(\sigma_l(e))$ is a uniquely restricted matching starting with $\sigma_l(e)$. Now, it follows from Observation~\ref{obs:esuccurm} and Corollary~\ref{cor:augment} that $\{e\}\cup U(\sigma_l(e))=U(e)$ is a uniquely restricted matching in $G$ starting with $e$. On the other hand, if (b) occurs, then $|U(\sigma_r(e))|=|U(e)|-1$, and therefore by the induction hypothesis, $U(\sigma_r(e))$ is a uniquely restricted matching starting with $\sigma_r(e)$. Then it follows from Observation~\ref{obs:esuccurm} and Corollary~\ref{cor:augment} that $\{e\}\cup U(\sigma_r(e))=U(e)$ is a uniquely restricted matching in $G$ starting with $e$. This completes the proof.
\end{proof}

\begin{lemma}\label{lem:algomax}
Let $M$ be a uniquely restricted matching starting with an edge $e\in E(G)$. Then, $|M|\leq |U(e)|$.
\end{lemma}
\begin{proof}
We prove this by induction on $|U(e)| = k$.

Assume that $k = 1$. In this case, $U(e) =\{e\}$.
Suppose that there exists a uniquely restricted matching $M$ starting with $e$ such that $|M|>1$.
Then there exists at least one edge $e'$ in $M$ such that $e'\neq e$.
Now since $U(e)=\{e\}$, we can see from the definition of $U(e)$ that neither $\sigma_l(e)$ nor $\sigma_r(e)$ exist. As we noted earlier, this means that $\rho(l(e))\geq v_{n-1}$ and $\rho(r(e))\geq v_n$. As $M$ starts with $e$ and $e'\in M$, we know by Lemma~\ref{lem:unitintdisjurm} that $l(e)<r(e)<l(e')<r(e')$. Since $\rho(l(e))\geq v_{n-1}$, it must be the case that $\rho(l(e))\geq l(e')$, which by Observation~\ref{obs:lambdarho} implies that $l(e)l(e')\in E(G)$. Similarly, since $\rho(r(e))\geq v_n$, we have $\rho(r(e))\geq r(e')$, from which it follows by Observation~\ref{obs:lambdarho} that $r(e)r(e')\in E(G)$. From Lemma~\ref{lem:unitinturmpair}, we now have that $\{e,e'\}$ is not a uniquely restricted matching, which is a contradiction to the fact that $\{e,e'\}\subseteq M$.
Therefore, the statement of the lemma is true when $k=1$.

Now, assume that $k>1$ and that the result is true for every edge $e'\in E(G)$ with $|U(e')|<k$.
Suppose that there exists a uniquely restricted matching $M$ starting with $e$ in $G$ such that $|M|>k$.
Let $U(e)=\{e=e_1, e_2,\ldots,e_k\}$ where $l(e_1)<l(e_2)<\cdots<l(e_k)$. From Lemma~\ref{lem:Uisurm}, we know that $U(e)$ is a uniquely restricted matching in $G$ and therefore by Lemma~\ref{lem:unitintdisjurm}, it follows that $l(e_1)<r(e_1)<l(e_2)<r(e_2)<\cdots<l(e_k)<r(e_k)$. Similarly let $M=\{e=e'_1,e'_2,\ldots,e'_{|M|}\}$ where $l(e'_1)<r(e'_1)<l(e'_2)<r(e'_2)<\cdots<l(e'_{|M|})<r(e'_{|M|})$. Note that as $k>1$ and $|M|>k$, the edges $e'_1$, $e'_2$ and $e'_3$ definitely exist.
\medskip

\noindent\textit{Claim 1.} Let $f\in\{\sigma_l(e),\sigma_r(e)\}$. Then either $l(f)>l(e'_2)$ or $r(f)>r(e'_2)$.

Suppose that this not the case. That is, for some $f\in\{\sigma_l(e),\sigma_r(e)\}$, we have $l(f)\leq l(e'_2)$ and $r(f)\leq r(e'_2)$. Then we can apply Lemma~\ref{lem:e1e2e3} to the edges $f$, $e'_2$ and $e'_3$ to conclude that $\{f,e'_3\}$ is a uniquely restricted matching in $G$. Note that as $l(f)\leq l(e'_2)<l(e'_3)$, this means by Lemma~\ref{lem:unitintdisjurm} that $r(f)<l(e'_3)$. Applying Corollary~\ref{cor:augment} to the edge $f$ and the uniquely restricted matching $\{e'_3,e'_4,\ldots,e'_{|M|}\}$, we can now infer that $M'=\{f,e'_3,e'_4,\ldots,e'_{|M|}\}$ is a uniquely restricted matching starting with $f$ in $G$. From the definition of $U(e)$, $|U(e)|\geq |U(f)|+1$. Therefore, we can infer that $|U(f)|\leq |U(e)|-1=k-1$. Then, by applying the induction hypothesis on $f$ and the matching $M'$, we have that $|M'|\leq |U(f)|\leq k-1$. As $|M'|=|M|-1$, we now have $|M|\leq k$, which is a contradiction to our assumption that $|M|>k$. This proves the claim.
\medskip

\noindent\textit{Claim 2.} $l(e)l(e'_2)\in E(G)$.

Suppose that $l(e)l(e'_2)\notin E(G)$. In this case, it is clear from Observation~\ref{obs:lambdarho} that $\rho(l(e))<l(e'_2)<r(e'_2)$. This implies that $\rho(l(e))<v_{n-1}$, telling us that $\sigma_l(e)$ exists, and that $l(\sigma_l(e))\leq l(e'_2)$. Note that as $l(\sigma_l(e))\leq l(e'_2)$, we have from the definition of $\sigma_l(e)$ that $r(\sigma_l(e))\leq r(e'_2)$. This contradicts Claim~1.
\medskip

As $\{e,e'_2\}$ is a uniquely restricted matching in $G$, we now have by Claim~2 and Lemma~\ref{lem:unitinturmpair} that $r(e)r(e'_2)\notin E(G)$. This implies by Observation~\ref{obs:lambdarho} that $\rho(r(e))<r(e'_2)$. Note that we now have $\rho(r(e))<v_n$, which means that $\sigma_r(e)$ exists. It is easy to see using the definition of $\sigma_r(e)$ that $r(\sigma_r(e))\leq r(e'_2)$. As $l(e)<r(e)<l(e'_2)$, we can deduce from Claim~2 and Definition~\ref{def:proporder} that $r(e)l(e'_2)\in E(G)$. Therefore, by Observation~\ref{obs:lambdarho}, $\rho(r(e))\geq l(e'_2)$. As $r(\sigma_r(e))>\rho(r(e))$, this gives us $l(e'_2)<r(\sigma_r(e))\leq r(e'_2)$. As $l(e'_2)r(e'_2)\in E(G)$, we now have by Definition~\ref{def:proporder} that $l(e'_2)r(\sigma_r(e))\in E(G)$, which implies by Observation~\ref{obs:lambdarho} that $\lambda(r(\sigma_r(e)))\leq l(e'_2)$. It can be seen from the definition of $\sigma_r(e)$ that $\lambda(r(\sigma_r(e)))=l(\sigma_r(e))$. Thus, we now have $l(\sigma_r(e))\leq l(e'_2)$. We again have a contradiction to Claim~1. This concludes the proof.
\end{proof}
\begin{remark}\label{rem:v1v2}
$U(v_1v_2)$ is a uniquely restricted matching of maximum cardinality in $G$.
\end{remark}
\begin{proof}
From Lemma~\ref{lem:Uisurm}, it is clear that $U(v_1v_2)$ is a uniquely restricted matching starting with $v_1v_2$ in $G$.
Let $\{e_1,e_2,\ldots,e_k\}$ be any uniquely restricted matching in $G$ where $l(e_1)<l(e_2)<\cdots<l(e_k)$. Clearly, $v_1\leq l(e_1)$ and $v_2\leq r(e_1)$. From Lemma~\ref{lem:unitintdisjurm}, we have $l(e_1)<r(e_1)<l(e_2)<r(e_2)<\cdots<l(e_k)<r(e_k)$. Therefore, we have $v_1\leq l(e_1)<l(e_2)$ and $v_2\leq r(e_1)<r(e_2)$. We can now apply Lemma~\ref{lem:e1e2e3} to the edges $v_1v_2$, $e_1$ and $e_2$ to conclude that $\{v_1v_2,e_2\}$ is a uniquely restricted matching in $G$. By Corollary~\ref{cor:augment}, we now have that $\{v_1v_2,e_2,e_3,\ldots,e_k\}$ is a uniquely restricted matching in $G$ starting with $v_1v_2$. As the cardinality of this matching is $k$, we have by Lemma~\ref{lem:algomax} that $|U(v_1v_2)|\geq k$.
\end{proof}
\begin{theorem}\label{thm:propintmaxurm}
There is a linear-time algorithm that computes a maximum cardinality uniquely restricted matching in a given proper interval graph.
\end{theorem}
\begin{proof}
Let the input graph $G$ have $n$ vertices and $m$ edges. We can assume that $G$ is connected, as if it is not, we can just run the algorithm separately in each component of $G$, and then return the union of the maximum cardinality uniquely restricted matchings found in each component.
From the input adjacency list, we can use the well known $O(n+m)$ time algorithms to generate a proper vertex ordering of $G$ (for example, \cite{Corneilproper} or \cite{HellHuang}). Assuming that the position of each vertex in the ordering is known as a unique integer in $[1,n]$ associated with each vertex, we can easily generate $l(e)$ and $r(e)$ for every edge $e\in E(G)$ in $O(n+m)$ time by making one pass through the adjacency list. During the same pass, we can compute $\lambda(u)$ and $\rho(u)$ for every vertex $u\in V(G)$. Once we are done with this, we can easily find $\sigma_l(e)$ and $\sigma_r(e)$ for any edge $e\in E(G)$ in $O(1)$ time. For every edge $e\in E(G)$, $U(e)$ can be stored as a list, which is empty to start with. The following subroutine computes $U(e)$ for a given edge $e$.

\noindent\rule{\textwidth}{.25mm}\\\textbf{Procedure} ComputeU($e$)\vspace{-.1in}\\\rule{\textwidth}{.25mm}

\begin{enumerate}
\itemsep 0in
\item \textbf{if} $\sigma_l(e)$ exists \textbf{and} $U(\sigma_l(e))=\emptyset$ \textbf{then} ComputeU($\sigma_l(e)$)
\item \textbf{if} $\sigma_r(e)$ exists \textbf{and} $U(\sigma_r(e))=\emptyset$ \textbf{then} ComputeU($\sigma_r(e)$)
\item \textbf{if} both $\sigma_l(e)$ and $\sigma_r(e)$ exist \textbf{then}
\item\hspace{.2in} \textbf{if} $|U(\sigma_l(e))|\geq |U(\sigma_r(e))|$ \textbf{then} set $U(e)=\{e\}\cup U(\sigma_l(e))$
\item\hspace{.2in} \textbf{else} set $U(e)=\{e\}\cup U(\sigma_r(e))$
\item \textbf{else if} $\sigma_l(e)$ exists \textbf{then} set $U(e)=\{e\}\cup U(\sigma_l(e))$
\item \textbf{else if} $\sigma_r(e)$ exists \textbf{then} set $U(e)=\{e\}\cup U(\sigma_r(e))$
\item \textbf{else} set $U(e)=\{e\}$
\end{enumerate}

The main algorithm just calls the procedure ComputeU($v_1v_2$), where $v_1$ and $v_2$ are the vertices at the first and second positions respectively in the ordering (recall that by Observation~\ref{obs:consecedge}, $v_1v_2\in E(G)$). The algorithm then finishes by returning the set of edges $U(v_1v_2)$. The correctness of the algorithm is guaranteed by Remark~\ref{rem:v1v2}. As every other part of the algorithm except the call to the procedure ComputeU($v_1v_2$) takes $O(n+m)$ time, we shall restrict our attention to the time taken to complete this call. Notice that for any edge $e\in E(G)$, the time spent inside the procedure ComputeU($e$) outside of the recursive calls to ComputeU is $O(1)$. Also observe that for any edge $e\in E(G)$, the call to ComputeU($e$) happens at most once, which means that there are at most $m$ calls to the procedure ComputeU. Therefore, the total time taken to complete the procedure ComputeU($v_1v_2$) is $O(m)$. Thus, the algorithm runs in time $O(n+m)$.
\end{proof}
\section{Bipartite Permutation Graphs}\label{sec:bipperm}
A bijection $\pi$ of $\{1, 2,\ldots,n\}$ to itself is called a permutation of order $n$. We write $\pi=(\pi(1),\pi(2),\ldots,\pi(n))$ to define a permutation of order $n$. The simple undirected graph $G_\pi$ associated with a permutation $\pi$ is a graph with $V(G_\pi)=\{1,2,\ldots,n\}$ and $E(G_\pi)=\{ij\colon (i-j)(\pi(i)-\pi(j))<0\}$. A graph $G$ on $n$ vertices is said to be a \emph{permutation graph} if it is isomorphic to $G_\pi$ for some permutation $\pi$ of order $n$. In other words, a graph $G$ on $n$ vertices is a permutation graph if there exists a bijection $f:V(G)\rightarrow\{1,2,\ldots,n\}$ and a permutation $\pi$ of order $n$ such that for $u,v\in V(G)$, $uv\in E(G)\Leftrightarrow (f(u)-f(v))(\pi(f(u))-\pi(f(v)))<0$.
Let $V(G)=\{v_1,v_2,\ldots,v_n\}$, where $v_i=f^{-1}(i)$ for $1\leq i\leq n$.

\begin{definition}
For a graph $G=(V,E)$, an ordering $<$ of $V(G)$ is said to be a \emph{transitive vertex ordering} if for $u,v,w\in V(G)$ such that $u<v<w$,
\begin{myenumeratea}
\item $uv,vw\in E(G)\Rightarrow uw\in E(G)$, and
\item $uw\in E(G)\Rightarrow uv\in E(G)$ or $vw\in E(G)$.
\end{myenumeratea}
\end{definition}

It is easy to see that $v_1,v_2,\ldots,v_n$ is a transitive vertex ordering of $G$.
(It is well known and easy to see that the digraph $\hat{G}$ with vertex set $V(G)$ and edge set $\{(u,v)\colon uv\in E(G)$ and $f(u)<f(v)\}$ is a partial order---in other words $\hat{G}$ corresponds to a transitive orientation of $G$. The bijection $f$ can be seen as a linear order of $V(G)$ that extends this partial order).
Using the fact that permutation graphs are exactly the graphs that are both comparability and co-comparability~\cite{DushnikMiller}, it is easy to see (and folklore) that a graph $G$ is a permutation graph if and only if it has a transitive vertex ordering (to see the sufficiency, observe that orienting every edge from left to right gives a transitive orientation of $G$ and orienting every non-edge from left to right gives a transitive orientation of $\overline{G}$, the complement of $G$).

A permutation graph that is also bipartite is called a \emph{bipartite permutation graph}. The class of bipartite permutation graphs is known to be the same as the classes of proper interval bigraphs, bipartite asteroidal-triple free graphs, bipartite co-comparability graphs and bipartite trapezoid graphs~\cite{HellHuang}. As permutation graphs do not contain odd induced cycles of length more than three, it is straightforward to verify that bipartite permutation graphs are exactly triangle-free permutation graphs. For such a graph, the following observation can easily be seen to be true.
\begin{observation}\label{obs:biptransitive}
Let $<$ be a transitive vertex ordering of a bipartite permutation graph $G$.
\begin{myenumeratea}
\item\label{it:now} If $u<v$ and $uv\in E(G)$, then there exists no $w>v$ such that $vw\in E(G)$.
\item\label{it:uvorvw} If $u<v<w$ and $uw\in E(G)$, then $uv\in E(G)$ or $vw\in E(G)$ but not both.
\end{myenumeratea}
\end{observation}

\renewcommand{\vertexset}{(3,0,0),(1,1,1),(2,0,2),(5,2,2),(6,2,0),(4,3,1),(7,4,2),(8,4,0)}
\renewcommand{\edgeset}{(2,1),(3,1),(5,1),(6,1),(5,4),(6,4),(4,7),(4,8)}
\renewcommand{\radius}{.1}
\begin{figure}[h!]
\centering
\begin{tikzpicture}
\drawgraph
\node at (-0.25,-0.25) {3};
\node at (0.75,1) {1};
\node at (-0.25,2.25) {2};
\node at (2,2.3) {5};
\node at (2,-0.3) {6};
\node at (3.25,1) {4};
\node at (4.25,2.3) {7};
\node at (4.25,-0.3) {8};
\end{tikzpicture}
\hspace{.3in}
\renewcommand{\vertexset}{(3,3,1),(1,1,1),(2,2,1),(5,5,1),(6,6,1),(4,4,1),(7,7,1),(8,8,1)}
\renewcommand{\edgeset}{(1,2,,,.25),(1,3,,,.5),(1,5,,,1),(1,6,,,1.5),(4,5,,,.25),(4,6,,,.5),(4,7,,,1),(4,8,,,1.5)}
\begin{tikzpicture}
\drawgraph
\node at (1,0.6) {1};
\node at (2,0.6) {2};
\node at (3,0.6) {3};
\node at (4,0.6) {4};
\node at (5,0.6) {5};
\node at (6,0.6) {6};
\node at (7,0.6) {7};
\node at (8,0.6) {8};
\end{tikzpicture}
\caption{The figure on the left shows a bipartite permutation graph where $f(u)$ is written near each vertex $u$ and $\pi=(2,3,5,6,1,7,8,4)$. The figure on the right shows a transitive vertex ordering of the graph on the left.}
\end{figure}
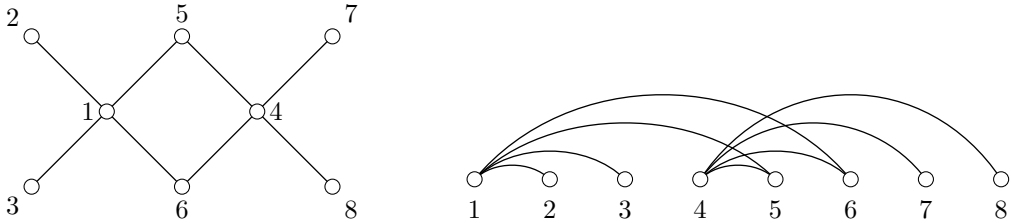

Let $G$ be a bipartite permutation graph with a transitive vertex ordering $<$.
For an edge $e=uv\in E(G)$, we define $l_<(e)=\min_<\{u,v\}$ and $r_<(e)=\max_<\{u,v\}$. When the ordering $<$ is clear from the context, we shorten $l_<(e)$ and $r_<(e)$ to $l(e)$ and $r(e)$ respectively.

A vertex $u\in V(G)$ is said to be a \emph{left-vertex} if there is some edge $e\in E(G)$ such that $u=l(e)$. Similarly, a vertex $u\in V(G)$ is said to be a \emph{right-vertex} if there exists some edge $e\in E(G)$ such that $u=r(e)$. As $G$ is a connected graph, it is clear that every vertex in $G$ is either a left-vertex or a right-vertex. We claim that no vertex can be both a left-vertex and a right-vertex. This is because if a vertex $u\in V(G)$ is such that $u=l(e)=r(e')$, for some $e,e'\in E(G)$, then we have $l(e')<u<r(e)$ and $l(e')u,ur(e)\in E(G)$, which contradicts Observation~\ref{obs:biptransitive}\ref{it:now}. Thus each vertex of $G$ is either a left-vertex or a right-vertex but not both.

If $u$ is a left-vertex, then it can have no neighbour $v$ such that $v<u$, because if it does, then $u$ becomes the right-vertex of the edge $vu$, which is a contradiction to the fact that no vertex can be both a left-vertex and a right-vertex. For the same reason, if $u$ is a right-vertex, it can have no neighbour $v$ such that $u<v$.

Since for every edge $e\in E(G)$, $l(e)$ is a left-vertex and $r(e)$ is a right-vertex, it can be concluded that every edge of $G$ is between a left-vertex and a right-vertex. This tells us that the set of left-vertices and the set of right-vertices are both independent sets of $G$. If $G$ has no isolated vertices, these sets form a bipartition of the bipartite graph $G$.

We say that a vertex $u\in V(G)$ is \emph{underneath} an edge $e\in E(G)$ if $l(e)\leq u\leq r(e)$. Suppose that a left-vertex $u\in V(G)$ is underneath an edge $e\in E(G)$. Clearly, $u\neq r(e)$, as it is a left-vertex and no vertex can be both a left-vertex and a right-vertex. If $u\neq l(e)$, it follows from Observation~\ref{obs:biptransitive}\ref{it:uvorvw} that $u$ is adjacent to $r(e)$ (note that $u$ cannot be adjacent to $l(e)$ as both are left-vertices). If $u=l(e)$, then clearly it is adjacent to $r(e)$. Therefore, we can conclude that if a left-vertex is underneath an edge $e\in E(G)$, then it is adjacent to $r(e)$. Using very similar arguments, we can also see that if a right-vertex is underneath an edge $e\in E(G)$, then it is adjacent to $l(e)$.

\begin{lemma}\label{lem:bippurmpair}
Let $G$ be a bipartite permutation graph with a transitive vertex ordering $<$ and let $e,e'\in E(G)$ such that $l(e)<l(e')$.
\begin{myenumeratea}
\item\label{it:r1l2} If $r(e)<l(e')$, then $\{e,e'\}$ is a uniquely restricted matching in $G$.
\item\label{it:l2r1r2} If $l(e')<r(e)<r(e')$, then $\{e,e'\}$ is a uniquely restricted matching in $G$ if and only if $l(e)r(e')\notin E(G)$.
\item\label{it:r2r1} If $r(e')<r(e)$, then $\{e,e'\}$ is not a uniquely restricted matching in $G$.
\end{myenumeratea}
\end{lemma}
\begin{proof}
Suppose that $r(e)<l(e')$. Then as we have $l(e)<r(e)<l(e')<r(e')$ and that $r(e)$ is a right-vertex, neither $l(e')$ nor $r(e')$ is a neighbour of $r(e)$. Then it can be seen from Lemma~\ref{lem:altc4} that there is no alternating cycle with respect to $\{e,e'\}$ in $G$, which further implies by Theorem~\ref{thm:urmaltcycle} that $\{e,e'\}$ is a uniquely restricted matching in $G$. This proves~\ref{it:r1l2}.

Now let us consider the case when $l(e')<r(e)<r(e')$. Then, as $r(e)$ is a right-vertex underneath the edge $e'$, we know that $l(e')r(e)\in E(G)$. If $l(e)r(e')\in E(G)$, then we have the alternating cycle $l(e)r(e')l(e')r(e)l(e)$ with respect to $\{e,e'\}$ in $G$. If $l(e)r(e')\notin E(G)$, then we have that there is no neighbour of $r(e')$ in $\{l(e),r(e)\}$ (note that $r(e)$ and $r(e')$ cannot be adjacent as they are both right-vertices), implying by Lemma~\ref{lem:altc4} that there is no alternating cycle with respect to $\{e,e'\}$ in $G$. This shows that there is an alternating cycle with respect to $\{e,e'\}$ in $G$ if an only if $l(e)r(e')\in E(G)$. This implies, by Theorem~\ref{thm:urmaltcycle}, that $\{e,e'\}$ is a uniquely restricted matching in $G$ if and only if $l(e)r(e')\notin E(G)$. This proves~\ref{it:l2r1r2}.

Finally, let us consider the case when $r(e')<r(e)$. Then we have $l(e)<l(e')<r(e')<r(e)$. Now, $l(e')$ is a left-vertex underneath $e$ and $r(e')$ is a right-vertex underneath $e$, implying that we have the edges $l(e)r(e'),l(e')r(e)\in E(G)$. Now we have the cycle $l(e)r(e')l(e')r(e)l(e)$ in $G$, which is an alternating cycle with respect to $\{e,e'\}$ in $G$. By Theorem~\ref{thm:urmaltcycle}, $\{e,e'\}$ is not a uniquely restricted matching in $G$. This completes the proof of~\ref{it:r2r1}.
\end{proof}

\begin{lemma}\label{lem:bipe1e2e3}
Let $G$ be a bipartite permutation graph with a transitive vertex ordering~$<$. Let $e_1,e_2,e_3\in E(G)$ such that $l(e_1)\leq l(e_2)\leq l(e_3)$ and $r(e_1)\leq r(e_2)\leq r(e_3)$. If $\{e_1,e_3\}$ is not a uniquely restricted matching in $G$ then neither $\{e_1,e_2\}$ nor $\{e_2,e_3\}$ is a uniquely restricted matching in $G$.
\end{lemma}
\begin{proof}
We shall show only that $\{e_2,e_3\}$ is not a uniquely restricted matching, but the same kind of reasoning can be used to show that $\{e_1,e_2\}$ is also not a uniquely restricted matching. If $\{e_2,e_3\}$ is not even a matching, then we are immediately done. So let us suppose that $\{e_2,e_3\}$ is a matching. Thus, we have $l(e_2)<l(e_3)$ and $r(e_2)<r(e_3)$, which implies that $l(e_1)<l(e_3)$ and $r(e_1)<r(e_3)$. As $\{e_1,e_3\}$ is not a uniquely restricted matching, Lemma~\ref{lem:bippurmpair}\ref{it:r1l2} tells us that we cannot have $r(e_1)<l(e_3)$. Therefore, it must be the case that $l(e_1)<l(e_3)<r(e_1)<r(e_3)$ (recall that $l(e_3)\neq r(e_1)$ as one is a left-vertex and the other a right-vertex), which by Lemma~\ref{lem:bippurmpair}\ref{it:l2r1r2} means that $l(e_1)r(e_3)\in E(G)$. Note that the previous inequality means that $l(e_1)\leq l(e_2)<l(e_3)<r(e_1)\leq r(e_2)<r(e_3)$. Thus, the left-vertex $l(e_2)$ is underneath the edge $l(e_1)r(e_3)$, which tells us that $l(e_2)r(e_3)\in E(G)$. We can now apply Lemma~\ref{lem:bippurmpair}\ref{it:l2r1r2} to the edges $e_2$ and $e_3$ to conclude that $\{e_2,e_3\}$ is not a uniquely restricted matching in $G$.
\end{proof}

We now show that the statement of Theorem~\ref{thm:urmintc4} also holds for bipartite permutation graphs.

\begin{theorem}\label{thm:urmbippc4}
Let $G=(V,E)$ be a bipartite permutation graph and let $M$ be a matching in it. Then, the following statements are equivalent:
\begin{myenumeratei}
\item\label{it:bipMurm} $M$ is uniquely restricted
\item\label{it:bipaltC4} There is no alternating cycle of length 4 with respect to $M$ in $G$
\item\label{it:bipee'urm} For any two edges $e,e'\in M$, $\{e,e'\}$ is a uniquely restricted matching in $G$
\end{myenumeratei}
\end{theorem}
\begin{proof}
As \ref{it:bipaltC4}$\Leftrightarrow$\ref{it:bipee'urm} is straightforward consequence of Lemma~\ref{lem:altc4} and Theorem~\ref{thm:urmaltcycle}, we shall only prove \ref{it:bipMurm}$\Leftrightarrow$\ref{it:bipaltC4}.

By Theorem~\ref{thm:urmaltcycle}, if $M$ is a uniquely restricted matching, then there is no alternating cycle of any length with respect to $M$ in $G$. So we only need to show that if $M$ is not uniquely restricted, then there is an alternating cycle of length 4 with respect to $M$ in $G$.

Let $<$ be a transitive vertex ordering of $G$. Suppose that the matching $M$ is not uniquely restricted. Then, by Theorem~\ref{thm:urmaltcycle}, there exists some alternating cycle with respect to $M$ in $G$. Let $u_1u_2u_3\ldots u_ku_1$ be an alternating cycle with respect to $M$ of smallest possible length in $G$. Clearly, $k$ is even and $4\leq k\leq |V(G)|$. If $k=4$, then this cycle is an alternating cycle of length 4 with respect to $M$ and we are done. So let us assume that $k>4$. We shall also assume without loss of generality that $u_1=\min_<\{u_1,u_2,\ldots,u_k\}$ and that $u_2<u_k$ (otherwise we can relabel the vertices of the cycle to satisfy both these conditions). Note that this means that $u_1$ is a left-vertex and that both $u_2$ and $u_k$ are right-vertices. Since set of left-vertices and set of right-vertices are both independent sets, we can see that $u_i$ is a left-vertex if and only if $i$ is odd. Now, let us examine the position of $u_3$ in the ordering $<$. As $u_3$ is a left-vertex and $u_2$ a neighbour of it, we must have $u_3<u_2$. As $u_2<u_k$, this means that $u_3$ is underneath the edge $u_1u_k$, which implies that $u_3u_k\in E(G)$. Then, at least one of the cycles $u_1u_2u_3u_ku_1$ or $u_3u_4u_5\ldots u_ku_3$ is an alternating cycle with respect to $M$ in $G$ having length smaller than $k$. This contradicts our assumption that $u_1u_2\ldots u_k$ is an alternating cycle with respect to $M$ of smallest possible length in $G$.
\end{proof}

Let $M$ be a matching in a bipartite permutation graph with a transitive vertex ordering $<$. The edges of $M$ can be labelled as $e_1,e_2,\ldots,e_{|M|}$ such that $l(e_1)<l(e_2)<\cdots<l(e_{|M|})$. The matching $M$ is then said to \emph{start with} the edge $e_1$.

\begin{theorem}\label{thm:bippconsec}
Let $G$ be a bipartite permutation graph with a transitive vertex ordering $<$.
Let $M=\{e_1, e_2,\ldots,$ $e_t\}$ be a matching in $G$ where $l(e_1)<l(e_2)<\cdots<l(e_t)$. The matching $M$ is uniquely restricted if and only if $\{e_i,e_{i+1}\}$ is a uniquely restricted matching in $G$, for each $i\in\{1,2,\ldots,t-1\}$.
\end{theorem}
\begin{proof}
The proof this theorem closely follows the proof of Theorem~\ref{thm:consecedges}.
As every subset of a uniquely restricted matching is also a uniquely restricted matching, to prove the theorem, we only need to show that whenever $M$ is not a uniquely restricted matching, there exists some $i\in\{1,2,\ldots,t-1\}$ such that $\{e_i,e_{i+1}\}$ is not a uniquely restricted matching. Suppose that $M$ is not a uniquely restricted matching. Then, by Theorem~\ref{thm:urmbippc4}, we know that there exists $e_p,e_q\in M$ with $1\leq p<q\leq t$ such that $\{e_p,e_q\}$ is not a uniquely restricted matching. We choose $p$ and $q$ such that for any $p',q'$ with $1\leq p'<q'\leq t$ and $q'-p'<q-p$, $\{e_{p'},e_{q'}\}$ is a uniquely restricted matching in $G$. Suppose that $q>p+1$. By our choice of $p$ and $q$, we know that both $\{e_p,e_{q-1}\}$ and $\{e_{q-1},e_q\}$ are uniquely restricted matchings. As we have $l(e_p)<l(e_{q-1})<l(e_q)$, we know by Lemma~\ref{lem:bippurmpair}\ref{it:r1l2} that $l(e_q)<r(e_p)$ and by Lemma~\ref{lem:bippurmpair}\ref{it:r2r1}, we have that $r(e_p)<r(e_{q-1})$ and $r(e_{q-1})<r(e_q)$. We now have $l(e_p)<l(e_{q-1})<l(e_q)<r(e_p)<r(e_{q-1})<r(e_q)$. By Lemma~\ref{lem:bippurmpair}\ref{it:l2r1r2}, we can now say that $l(e_p)r(e_q)\in E(G)$ and that $l(e_{q-1})r(e_q)\notin E(G)$. But this is impossible as $l(e_{q-1})$ is now a left-vertex underneath the edge $l(e_p)r(e_q)$. Therefore, we can conclude that $q=p+1$. For $i=p$, we now have that $\{e_i,e_{i+1}\}$ is not a uniquely restricted matching, thereby completing the proof.

\end{proof}
\begin{corollary}\label{cor:bippaugment}
Let $M$ be a uniquely restricted matching in a bipartite permutation graph $G$ starting with $e'\in E(G)$. Let $e\in E(G)$ such that $l(e)<l(e')$ and $\{e,e'\}$ is a uniquely restricted matching in $G$. Then $\{e\}\cup M$ is a uniquely restricted matching in $G$ starting with $e$.
\end{corollary}
\begin{proof}
We shall first show that $\{e\}\cup M$ is a matching. As $\{e,e'\}$ is a uniquely restricted matching, it follows from Lemma~\ref{lem:bippurmpair}\ref{it:r2r1} that $r(e)<r(e')$. For any edge $e''\in M\setminus\{e'\}$, we have $l(e')<l(e'')$ as $M$ starts with $e'$ and therefore, from Lemma~\ref{lem:bippurmpair}\ref{it:r2r1} and the fact that $\{e',e''\}$ is a uniquely restricted matching, we have $r(e')<r(e'')$. This tells us that for every edge $e''\in M$, $r(e)<r(e'')$. Note that we also have $l(e)<l(e'')$ for every edge $e''\in M$. Then, $l(e)$ and $r(e)$ are distinct from $l(e'')$ and $r(e'')$ for any edge $e''\in M$ (recall that no vertex can be both a left-vertex and a right-vertex). This leads us to the conclusion that $\{e\}\cup M$ is a matching. The proof of the corollary now follows directly from Theorem~\ref{thm:bippconsec}.
\end{proof}

From here onwards, we assume that $G$ is a bipartite permutation graph with no isolated vertices and having a transitive vertex ordering~$<$.
Let $V(G)=\{v_1,v_2,\ldots,v_n\}$ where $v_1<v_2<\cdots<v_n$.
For a vertex $u\in V(G)$, as in Section~\ref{sec:propint}, we define $\lambda(u)=\min_<\{v\in N(u)\cup\{u\}\}$ and $\rho(u)=\max_<\{v\in N(u)\cup\{u\}\}$. For each left-vertex $u$, we further define $\gamma(u)=\min_<\{v\in N(u)\colon u<v\}$.

For every edge $e\in E(G)$, we now define a pair of edges $x(e)$ and $y(e)$ as follows.
Let $e\in E(G)$ such that $\rho(l(e))=v_i$. Then, we define $x(e)=\lambda(u)u$, where
$$u=\left\{\begin{array}{cl}v_{i+1}&\mbox{if }v_{i+1}\mbox{ is a right-vertex}\\\gamma(v_{i+1})&\mbox{otherwise}\end{array}\right.$$
It is easy to see that $x(e)$ does not exist if and only if $\rho(l(e))=v_n$ (recall that $G$ has no isolated vertices).
We define $y(e)=u\gamma(u)$ where $u=\min_<\{v\in V(G)\colon v>r(e)$ and $v$ is a left-vertex$\}$. Clearly, $y(e)$ does not exist if and only if either $r(e)=v_n$ or if every $v>r(e)$ is a right-vertex (again, recall that $G$ has no isolated vertices). 

\begin{observation}\label{obs:exyurm}
Let $e\in E(G)$.
\begin{myenumeratea}
\item\label{it:exurm} If $x(e)$ exists, then $l(e)<l(x(e))$ and $\{e,x(e)\}$ is a uniquely restricted matching in $G$.
\item\label{it:eyurm} If $y(e)$ exists, then $l(e)<l(y(e))$ and $\{e,y(e)\}$ is a uniquely restricted matching in $G$.
\end{myenumeratea}
\end{observation}
\begin{proof}
We shall first prove~\ref{it:exurm}. By definition of $x(e)$, we know that $\rho(l(e))<r(x(e))$. As $r(e)\leq\rho(l(e))$, this implies that $l(e)<r(e)<r(x(e))$ and that $l(e)r(x(e))\notin E(G)$. As $l(x(e))r(x(e))\in E(G)$, this means that $l(x(e))\neq l(e)$. If $l(x(e))<l(e)$, then as we have $l(x(e))<l(e)<r(x(e))$, the left-vertex $l(e)$ underneath the edge $l(x(e))r(x(e))$ has to be adjacent to $r(x(e))$, contradicting our previous observation. Therefore, we can conclude that $l(e)<l(x(e))$. If $r(e)<l(x(e))$, then by Lemma~\ref{lem:bippurmpair}\ref{it:r1l2}, we have that $\{e,x(e)\}$ is a uniquely restricted matching in $G$, and thus we are done. So let us assume that $l(x(e))<r(e)$ (note that they cannot be equal as one is left-vertex and the other a right-vertex). We now have the inequality $l(e)<l(x(e))<r(e)<r(x(e))$. Now since $l(e)r(x(e))\notin E(G)$, Lemma~\ref{lem:bippurmpair}\ref{it:l2r1r2} can be used to conclude that $\{e,x(e)\}$ is a uniquely restricted matching in $G$. This completes the proof of~\ref{it:exurm}.

Next, we shall prove~\ref{it:eyurm}. From the definition of $y(e)$, it is clear that $r(e)<l(y(e))$ and therefore $l(e)<l(y(e))$. Furthermore, from Lemma~\ref{lem:bippurmpair}\ref{it:r1l2}, we get that $\{e,y(e)\}$ is a uniquely restricted matching in $G$, thus proving~\ref{it:eyurm}.
\end{proof}
\medskip

We shall now define a set $U(e)$ of edges for every edge $e\in E(G$) as follows.
$$U(e)=\left\{\begin{array}{cl}\{e\}&\mbox{if neither }x(e)\mbox{ nor }y(e)\mbox{ exists}\\\{e\}\cup U(x(e))&\mbox{if }y(e)\mbox{ does not exist, or if both exist and }|U(x(e))|\geq |U(y(e))|\\\{e\}\cup U(y(e))&\mbox{if }x(e)\mbox{ does not exist, or if both exist and }|U(x(e))|<|U(y(e))|\end{array}\right.$$
From Observation~\ref{obs:exyurm}, we know that $l(e)<l(x(e))$ and $l(e)<l(y(e))$, which implies that $U(e)$ is well-defined. The next two lemmas will show that $U(e)$ is always a uniquely restricted matching starting with $e$ and that it has the maximum possible cardinality among all the uniquely restricted matchings starting with $e$ in $G$.

\begin{lemma}\label{lem:bipUisurm}
For any edge $e\in E(G)$, $U(e)$ is a uniquely restricted matching starting with $e$ in $G$.
\end{lemma}
\begin{proof}
We shall prove this by induction on $|U(e)|$. If $|U(e)|=1$, then it must be the case that $U(e)=\{e\}$. In this case, the statement of the lemma is clearly true. Now let us assume that $|U(e)|>1$ and that the statement of the lemma is true for all $e'\in E(G)$ such that $|U(e')|<|U(e)|$. As $|U(e)|>1$, we know that at least one of $x(e),y(e)$ exists. From the definition of $U(e)$, it can be seen there are only two possibilities: $x(e)$ exists and $U(e)=\{e\}\cup U(x(e))$, or $y(e)$ exists and $U(e)=\{e\}\cup U(y(e))$. From the induction hypothesis, $U(x(e))$ is a uniquely restricted matching starting with $x(e)$ and $U(y(e))$ is a uniquely restricted matching starting with $y(e)$. It now follows from Observation~\ref{obs:exyurm} and Corollary~\ref{cor:bippaugment} that $U(e)$ is a uniquely restricted matching starting with $e$.
\end{proof}

\begin{lemma}\label{lem:bipalgomax}
Let $M$ be a uniquely restricted matching starting with $e$ in $G$. Then $|M|\leq |U(e)|$.
\end{lemma}
\begin{proof}
We will use induction on $|U(e)|=k$ to prove this. Suppose first that $k=1$. Then $U(e)=\{e\}$, which can be the case only if neither $x(e)$ nor $y(e)$ exist. As $x(e)$ does not exist, we have $\rho(l(e))=v_n$. In this case, for an edge $e'\neq e$ in $M$, we must have $r(e')$ underneath the edge $l(e)\rho(l(e))$, implying that $l(e)r(e')\in E(G)$. As $\{e,e'\}$ is a uniquely restricted matching, from Lemmas~\ref{lem:bippurmpair}\ref{it:l2r1r2} and~\ref{lem:bippurmpair}\ref{it:r2r1}, we have $r(e)<l(e')$. This tells us that $r(e)\neq v_n$ and that there exists a left-vertex $l(e')>r(e)$. But this contradicts the fact that $y(e)$ does not exist. We can therefore conclude that $e'$ does not exist, or in other words, $M=\{e\}$, thereby proving the statement of the lemma. We shall now assume that $k>1$ and that for any edge $e'\in E(G)$ such that $|U(e')|<k$, the statement of the lemma is true.

Let us assume for the sake of contradiction that $|M|>k$.
From Lemma~\ref{lem:bipUisurm}, we know that $U(e)$ is a uniquely restricted matching starting with $e$ in $G$. Let $U(e)=\{e=e_1,e_2,\ldots,e_k\}$ such that $l(e_1)<l(e_2)<\cdots<l(e_k)$. Let $M=\{e=e'_1,e'_2,\ldots,e'_{|M|}\}$ such that $l(e'_1)<l(e'_2)<\cdots<l(e'_{|M|})$. Note that as $k>1$, we have $|M|>2$, implying that the edges $e'_1$, $e'_2$ and $e'_3$ definitely exist.
\medskip

\noindent\textit{Claim 1.} Let $f\in\{x(e),y(e)\}$. Then either $l(f)>l(e'_2)$ or $r(f)>r(e'_2)$.

Suppose that for some $f\in\{x(e),y(e)\}$, we have $l(f)\leq l(e'_2)$ and $r(f)\leq r(e'_2)$. Then we have $l(f)\leq l(e'_2)<l(e'_3)$ and $r(f)\leq r(e'_2)<r(e'_3)$. Applying Lemma~\ref{lem:bipe1e2e3} to the edges $f$, $e'_2$ and $e'_3$, we now get that $\{f,e'_3\}$ is a uniquely restricted matching. As $\{e'_3,e'_4,\ldots,e'_{|M|}\}$ is a uniquely restricted matching starting with $e'_3$, we can now use Corollary~\ref{cor:bippaugment} to conclude that $M'=\{f,e'_3,e'_4,\ldots,e'_{|M|}\}$ is a uniquely restricted matching in $G$ starting with $f$. From the definition of $U(e)$, we can see that $|U(e)|\geq |U(f)|+1$. Therefore, we have $|U(f)|\leq k-1$. Then, by applying the induction hypothesis on $f$ and $M'$, we have $|M'|\leq |U(f)|$. As $|M'|=|M|-1$, this means that $|M|-1\leq |U(f)|\leq k-1$. We thus have $|M|\leq k$, which is a contradiction to our assumption that $|M|>k$. This proves the claim.
\medskip

Suppose that $l(e'_2)<r(e)$. Since $\{e,e'_2\}$ is a uniquely restricted matching, we have from Lemma~\ref{lem:bippurmpair}\ref{it:r2r1} that $r(e)<r(e'_2)$ and then further from Lemma~\ref{lem:bippurmpair}\ref{it:l2r1r2} that $l(e)r(e'_2)\notin E(G)$. This means that $r(e'_2)$ cannot be underneath the edge $l(e)\rho(l(e))$, which implies that $\rho(l(e))<r(e'_2)$. Let $\rho(l(e))=v_i$. Then, $v_{i+1}\leq r(e'_2)$. Now suppose that $r(x(e))>r(e'_2)$. Then clearly, $r(x(e))\neq v_{i+1}$. This can only mean that $v_{i+1}$ is a left-vertex, which implies that $v_{i+1}<r(e'_2)$, and that $r(x(e))=\gamma(v_{i+1})$, which implies that $\gamma(v_{i+1})>r(e'_2)$ (as we have assumed that $r(x(e))>r(e'_2)$).
Since we now have $v_{i+1}<r(e'_2)<\gamma(v_{i+1})$, the vertex $r(e'_2)$ is a right-vertex underneath the edge $v_{i+1}\gamma(v_{i+1})$, implying that $v_{i+1}r(e'_2)\in E(G)$. But this contradicts our choice of $\gamma(v_{i+1})$. Therefore, we can conclude that $r(x(e))\leq r(e'_2)$. It is easy to see that by the definition of $x(e)$, we always have $r(e)\leq\rho(l(e))<r(x(e))$. As $l(e'_2)<r(e)$, we now have $l(e'_2)<r(x(e))\leq r(e'_2)$. Then, $r(x(e))$ is a right-vertex underneath the edge $e'_2$, which means that $l(e'_2)r(x(e))\in E(G)$. Since by definition of $x(e)$, we have $l(x(e))=\lambda(r(x(e)))$, this implies that $l(x(e))\leq l(e'_2)$. Since we also have $r(x(e))\leq r(e'_2)$, we now have a contradiction to Claim~1.

The only remaining case is when $r(e)<l(e'_2)$. As $l(e'_2)$ is a left-vertex that comes after $r(e)$ in the ordering $<$, we know that $y(e)$ exists and that $l(y(e))\leq l(e'_2)$. It can be seen from the definition of $y(e)$ that $r(y(e))=\gamma(l(y(e)))$. If $r(y(e))>r(e'_2)$, then we have $l(y(e))\leq l(e'_2)<r(e'_2)<r(y(e))$, and therefore $r(e'_2)$ is a right-vertex underneath the edge $y(e)$, which implies that $l(y(e))r(e'_2)\in E(G)$. But this contradicts the earlier observation that $r(y(e))=\gamma(l(y(e)))$. We can therefore conclude that $r(y(e))\leq r(e'_2)$. Recalling that $l(y(e))\leq l(e'_2)$, we now have a contradiction to Claim~1. This completes the proof.
\end{proof}

\begin{remark}\label{rem:v1gammav1}
$U(v_1\gamma(v_1))$ is a uniquely restricted matching of maximum cardinality in $G$.
\end{remark}
\begin{proof}
From Lemma~\ref{lem:bipUisurm}, it is clear that $U(v_1\gamma(v_1))$ is a uniquely restricted matching starting with $v_1\gamma(v_1)$ in $G$. Let $\{e_1,e_2,\ldots,e_k\}$ be any uniquely restricted matching in $G$ where $l(e_1)<l(e_2)<\cdots<l(e_k)$. Clearly, $v_1\leq l(e_1)<l(e_2)$. As $\{e_1,e_2\}$ is a uniquely restricted matching, we know from Lemma~\ref{lem:bippurmpair}\ref{it:r2r1} that $r(e_1)<r(e_2)$. If $r(e_1)<\gamma(v_1)$, then $r(e_1)$ is a right-vertex underneath the edge $v_1\gamma(v_1)$, implying that $v_1r(e_1)\in E(G)$. But this would be a contradiction to the choice of $\gamma(v_1)$. It must therefore be the case that $\gamma(v_1)\leq r(e_1)$. We now have $v_1\leq l(e_1)<l(e_2)$ and $\gamma(v_1)\leq r(e_1)<r(e_2)$. We can now apply Lemma~\ref{lem:bipe1e2e3} to the edges $v_1\gamma(v_1)$, $e_1$ and $e_2$ to conclude that $\{v_1\gamma(v_1),e_2\}$ is a uniquely restricted matching in $G$. By Corollary~\ref{cor:bippaugment}, we now have that $\{v_1\gamma(v_1),e_2,e_3,\ldots,e_k\}$ is a uniquely restricted matching in $G$ starting with $v_1\gamma(v_1)$. As the cardinality of this matching is $k$, we have by Lemma~\ref{lem:bipalgomax} that $|U(v_1\gamma(v_1))|\geq k$.
\end{proof}

\begin{theorem}
There is a linear-time algorithm that given a bipartite permutation graph as input, computes a maximum cardinality uniquely restricted matching in it.
\end{theorem}
\begin{proof}
We can construct a linear-time algorithm along the lines of the proof of Theorem~\ref{thm:propintmaxurm}. We first remove isolated vertices from $G$ and then use one of the known linear-time algorithms to generate a transitive vertex ordering $<$ of $V(G)$ (for example,~\cite{HellHuang}). In a single pass through the adjacency list that takes time $O(n+m)$, every vertex in $G$ can be marked as a left-vertex or right-vertex and the values $\lambda(u)$, $\rho(u)$ for each vertex $u\in V(G)$ and the value $\gamma(u)$ for each left-vertex $u$ can be computed. The algorithm further computes for all $u\in V(G)$, a value $\nu(u)=\min_<\{v\in V(G)\colon v>u$ and $v$ is a left-vertex$\}$ using a single pass in the backward direction through the vertex ordering $<$, taking $O(n)$ time. It is not hard to see that once this is done, the values $x(e)$ and $y(e)$ for an edge $e\in E(G)$ can be computed in $O(1)$ time. Then, a dynamic programming algorithm very similar to the one from the proof of Theorem~\ref{thm:propintmaxurm} can be used to compute $U(e)$ for every edge $e\in E(G)$, in $O(n+m)$ time. Finally, the algorithm returns $U(v_1\gamma(v_1))$. The correctness of the algorithm follows from Lemma~\ref{lem:bipUisurm}, Lemma~\ref{lem:bipalgomax} and Remark~\ref{rem:v1gammav1} and the algorithm clearly runs in $O(n+m)$ time.
\end{proof}

\section{Interval Graphs}\label{sec:int}
In this section, we present a polynomial-time algorithm that computes a maximum cardinality uniquely restricted matching in a given interval graph. Our approach will be to reduce the problem to the problem of finding a maximum cardinality ``strong independent set'' in an interval nest digraph, whose interval nest representation is given.
\subsection{The Strong Independent Set Problem in Interval Nest Digraphs}
For directed graphs, the term ``interval representation'' has a different meaning. Given a digraph $G$, a collection $\{(S_u,T_u)\}_{u\in V(G)}$ of pairs of intervals is said to be an interval representation of $G$ if for distinct $u,v\in V(G)$, we have $(u,v)\in E(G)$ if and only if $S_u\cap T_v\neq\emptyset$. The directed graphs which have interval representations are called \emph{interval digraphs}~\cite{Senetal}. If the collection $\{(S_u,T_u)\}_{u\in V(G)}$ has the property that $T_u\subseteq S_u$ for every $u\in V(G)$, then it is called an \emph{interval nest representation} of $G$. The digraphs that have interval nest representations are called \emph{interval nest digraphs}~\cite{Prisner}.

For a directed graph $G$, a set of vertices $S\subseteq V(G)$ is said to be a \emph{strong independent set} if for any two vertices $u,v\in S$, either $(u,v)\notin E(G)$ or $(v,u)\notin E(G)$.

In this section, we present a polynomial-time dynamic programming algorithm to compute a maximum cardinality strong independent set in an interval nest digraph $G$, whose interval nest representation is given.
Note that we can assume that the endpoints of all the intervals in the interval nest representation are distinct---otherwise, we can slightly perturb the endpoints of the intervals to obtain a new interval nest representation of the digraph in which this is true. Note also that we can assume every interval in the interval nest representation has integer endpoints. As there are four endpoints corresponding each vertex, we can assume that each endpoint in the representation is a unique integer in $[1,4|V(G)|]$. 

Let the interval nest representation of the input interval nest digraph $G$ be $\{(S_u,T_u)\}_{u\in V(G)}$. For a vertex $u\in V(G)$, let $S_u=[L_u,R_u]$ and $T_u=[l_u,r_u]$. Because of our assumptions about the interval nest representation, we have $L_u<l_u<r_u<R_u$.

For any vertex $x\in V(G)$, let $\eta(x)$ denote the vertex such that $l_x<l_{\eta(x)}$ but there does not exist any vertex $x'\in V(G)$ such that $l_x<l_{x'}<l_{\eta(x)}$. 

For vertices $u,v\in V(G)$, we define
$$X(u,v)=\left\{\begin{array}{cl}\{y\in V(G)\colon r_u<L_y<R_y<l_v\}&\mbox{when }r_u<l_v\mbox{ and }L_v<r_u,\\
\emptyset&\mbox{otherwise}\end{array}\right.$$

In addition, for $u,v,x\in V(G)$, define $$Y(u,v,x)=\left\{\begin{array}{cl}\{y\in X(u,v)\colon l_y\geq l_x\}&\mbox{when }r_u<l_x<l_v\\\emptyset&\mbox{otherwise}\end{array}\right.$$

Note that $X(u,v)=Y(u,v,\eta(u))$.
We shall now define our dynamic programming table $S$ in which there is an entry $S(u,v,x)\subseteq Y(u,v,x)$ for every triple of vertices $(u,v,x)\in V(G)^3$ ($=V(G)\times V(G)\times V(G)$).
Note that by our definition of $X(u,v)$ and $Y(u,v,x)$, the entry $S(u,v,x)$ corresponding to the triple $(u,v,x)$ will be $\emptyset$ if at least one of the conditions $r_u<l_v$, $L_v<r_u$ or $r_u<l_x<l_v$ is not true. We shall ensure that $S(u,v,x)$ is a strong independent set of maximum possible cardinality among all the strong independent sets that contain only vertices in $Y(u,v,x)$. In other words, $S(u,v,x)$ is a maximum cardinality strong independent set in the subdigraph induced in $G$ by $Y(u,v,x)$.

We give below the pseudocode for a procedure that computes $S(u,v,x)$, given $u,v,x\in V(G)$.
\medskip

\noindent\rule{\textwidth}{.25mm}\\\textbf{Procedure} ComputeS($u,v,x$)\vspace{-.1in}\\\rule{\textwidth}{.25mm}

\begin{enumerate}
\itemsep 0in
\item \textbf{if} $Y(u,v,x)=\emptyset$ \textbf{then}
\item \hspace{.2in} set $S(u,v,x)=\emptyset$
\item \hspace{.2in} \textbf{return}
\item\label{step:assign1} Set $T=S(u,v,\eta(x))$
\item \textbf{if} $x\in X(u,v)$ \textbf{then}
\item \hspace{.2in} Set $T'=\{x\}\cup S(x,v,\eta(x))$
\item\label{step:assign2} \hspace{.2in} \textbf{if} $|T'|>|T|$ \textbf{then} set $T=T'$
\item \hspace{.2in} Set $B=\{y\in Y(u,v,x)\colon L_y<r_x\mbox{ and }R_x<l_y\}$
\item \hspace{.2in} \textbf{for each} $y\in B$,
\item \hspace{.4in} Set $T'=\{x\}\cup S(x,y,\eta(x))\cup S(u,v,y)$
\item\label{step:assign3} \hspace{.4in} \textbf{if} $|T'|>|T|$ \textbf{then} set $T=T'$
\item Set $S(u,v,x)=T$
\end{enumerate}

\begin{lemma}\label{lem:issis}
For $(u,v,x)\in V(G)^3$, $S(u,v,x)$ is a subset of $Y(u,v,x)$ and is a strong independent set in $G$.
\end{lemma}
\begin{proof}
We shall prove this by induction on $l_v-l_x$. If $l_v-l_x\leq 0$, then we have $Y(u,v,x)=\emptyset$ and therefore $S(u,v,x)=\emptyset$. Clearly, the statement of the lemma is true in this case. Now let us assume that the statement has been proved for all $(u',v',x')\in V(G)^3$ such that $l_{v'}-l_{x'}<l_v-l_x$.

If $S(u,v,x)=\emptyset$, then there is nothing to prove.
Otherwise, it is the set that got assigned to $T$ in the last step where the value of $T$ was changed. This last step where $T$'s value was changed might be step~\ref{step:assign1}, step~\ref{step:assign2} or an iteration of step~\ref{step:assign3}.
Moreover, $Y(u,v,x)\neq\emptyset$ from which it follows that $l_x>r_u$.

First, let us consider the case when
the last time $T$ got assigned was in step~\ref{step:assign1}. In this case,
$S(u,v,x)=S(u,v,\eta(x))$. As $l_v-l_x>l_v-l_{\eta(x)}$, we can use the induction hypothesis to conclude that $S(u,v,\eta(x))\subseteq Y(u,v,\eta(x))$. Since $Y(u,v,\eta(x))\subseteq Y(u,v,x)$ (recall that $l_x>r_u$), we have $S(u,v,x)=S(u,v,\eta(x))\subseteq Y(u,v,x)$. It is immediately clear from the induction hypothesis that $S(u,v,\eta(x))=S(u,v,x)$ is a strong independent set in $G$.

Next, we consider the case when
the last time $T$ got assigned a set was in step~\ref{step:assign2}. Then, we know that $x\in X(u,v)$ which
implies that $x\in Y(u,v,x)$. We also have $S(u,v,x)=\{x\}\cup S(x,v,\eta(x))$. Again, by the induction hypothesis, we have $S(x,v,\eta(x))\subseteq Y(x,v,\eta(x))$ and that $S(x,v,\eta(x))$ is a strong independent set in $G$. As $x\in X(u,v)$, we have $Y(x,v,\eta(x))\subseteq Y(u,v,x)$. Since we also have $x\in Y(u,v,x)$, we can conclude that $S(u,v,x)=(\{x\}\cup S(x,v,\eta(x)))\subseteq Y(u,v,x)$. To see that $\{x\}\cup S(x,v,\eta(x))$ is a strong independent set in $G$, observe that for every vertex $w\in S(x,v,\eta(x))\subseteq Y(x,v,\eta(x))=X(x,v)$, we have $r_x<L_w$, implying that $(w,x)\notin E(G)$.

Finally, consider the case when the last time that an assignment to $T$ took place was in an iteration of step~\ref{step:assign3}. Again, it must be the case that $x\in X(u,v)$, which implies that $x\in Y(u,v,x)$. Also, we have $S(u,v,x)=\{x\}\cup S(x,y,\eta(x))\cup S(u,v,y)$ for some $y\in B\subseteq Y(u,v,x)\subseteq X(u,v)$. By the induction hypothesis, we have $S(x,y,\eta(x))\subseteq Y(x,y,\eta(x))$ and $S(u,v,y)\subseteq Y(u,v,y)$. As we have $x,y\in X(u,v)$, we can conclude that $Y(x,y,\eta(x))\subseteq Y(u,v,x)$ and thereby $S(x,y,\eta(x))\subseteq Y(u,v,x)$. From the definition of $B$, it is clear that $l_x<l_y$, implying that $Y(u,v,y)\subseteq Y(u,v,x)$, and therefore $S(u,v,y)\subseteq Y(u,v,x)$. Altogether, we now have $S(u,v,x)=(\{x\}\cup S(x,y,\eta(x))\cup S(u,v,y))\subseteq Y(u,v,x)$. It only remains to be shown that $S(u,v,x)=\{x\}\cup S(x,y,\eta(x))\cup S(u,v,y)$ is a strong independent set in $G$. It is easy to see that for every vertex $w\in S(x,y,\eta(x))\subseteq Y(x,y,\eta(x))=X(x,y)$, we have $r_x<L_w$ and therefore, $(w,x)\notin E(G)$. Now consider a vertex $w'\in S(u,v,y)\subseteq Y(u,v,y)$. Clearly, $l_{w'}\geq l_y$. From the definition of $B$, we have $R_x<l_y$ which now gives us $R_x<l_{w'}$. This means that $(x,w')\notin E(G)$. Finally, let us consider a vertex $w\in S(x,y,\eta(x))\subseteq Y(x,y,\eta(x))=X(x,y)$ and a vertex $w'\in S(u,v,y)\subseteq Y(u,v,y)$. Clearly, $R_w<l_y\leq l_{w'}$, implying that $(w,w')\notin E(G)$.
\end{proof}

\begin{lemma}\label{lem:ismax}
Let $(u,v,x)\in V(G)^3$ and let $S'\subseteq Y(u,v,x)$ be a strong independent set in $G$. Then $|S'|\leq |S(u,v,x)|$.
\end{lemma}
\begin{proof}
We shall prove this by induction on $l_v-l_x$. If $l_v-l_x\leq 0$, then we have $Y(u,v,x)=\emptyset$ and therefore $S(u,v,x)=\emptyset$. Clearly, the statement of the lemma is true in this case. Now let us assume that the statement has been proved for all $(u',v',x')\in V(G)^3$ such that $l_{v'}-l_{x'}<l_v-l_x$.

First let us note that the procedure ComputeS($u,v,x$) actually computes $S(u,v,x)$ as given by the following expression, where $\mathrm{Max}(\mathcal{F})$ denotes a set of maximum cardinality in a family $\mathcal{F}$ of sets.
\medskip

\begin{equation}\label{eqn:max}
S(u,v,x)=\left\{\begin{array}{cl}S(u,v,\eta(x))&\mbox{if }x\notin X(u,v)\\\\
\mathrm{Max}\left(\begin{array}{c}\left\{S(u,v,\eta(x)),\{x\}\cup S(x,v,\eta(x))\right\}\\
\cup\\\left\{\{x\}\cup S(x,y,\eta(x))\cup S(u,v,y)\colon y\in B\right\}\end{array}\right)&\mbox{if }x\in X(u,v)\end{array}\right.
\end{equation}
\medskip

Let $S'\subseteq Y(u,v,x)$ be a strong independent set in $G$. Let us first consider the case in which $x\notin S'$. In this case, it is easy to see that $S'\subseteq Y(u,v,\eta(x))$. From the induction hypothesis, we have $|S'|\leq |S(u,v,\eta(x))|$. It follows from equation~(\ref{eqn:max}) that $|S(u,v,x)|\geq |S(u,v,\eta(x))|$ and therefore we are done.

Now let us consider the case when $x\in S'$. Note that since $S'\subseteq Y(u,v,x)\subseteq X(u,v)$, we now have $x\in X(u,v)$.

Suppose first that there exists some vertex $z\in S'\setminus\{x\}$ such that $L_z<r_x$. Then let $z$ be that vertex in $S'\setminus\{x\}$ with $L_z<r_x$ such that there exists no vertex $z'\in S'\setminus\{x\}$ with $l_{z'}<l_z$ and $L_{z'}<r_x$. Let $S'_1=S'\cap X(x,z)$ and $S'_2=S'\setminus (\{x\}\cup S'_1)$. Note that $S'$ is a disjoint union of the sets $\{x\}$, $S'_1$ and $S'_2$ and that $z\in S'_2$. We claim that for each vertex $w\in S'_2$, we have $l_w\geq l_z$. Suppose that there exists $w\in S'_2$ such that $l_w<l_z$. As $S'$ is a strong independent set containing both $w$ and $z$, it must be the case that $r_w<l_z$ (otherwise, $[l_w,r_w]\cap [l_z,r_z]\neq\emptyset$, implying that both $(w,z),(z,w)\in E(G)$). If $L_w<r_x$, then we have a contradiction to our choice of $z$. Therefore, we have $r_x<L_w$. Recalling that $L_z<r_x$, we now have $L_z<r_x<L_w<r_w<l_z$. Then, the only reason $w\notin X(x,z)$ must be the fact that $l_z<R_w$. But now we have $L_z<r_w<l_z<R_w$, implying that both $(w,z),(z,w)\in E(G)$. But this is impossible as both $z$ and $w$ belong to a strong independent set $S'$ of $G$. This allows us to conclude that every vertex $w\in S'_2$ has the property that $l_w\geq l_z$. Therefore, recalling that $S'\subseteq X(u,v)$, we can infer that $S'_2\subseteq Y(u,v,z)$. Clearly, $S'_1\subseteq X(x,z)=Y(x,z,\eta(x))$. Since $S'\subseteq Y(u,v,x)$ and $z\in S'\setminus\{x\}$, we have $l_x<l_z<R_z<l_v$, implying that $l_v-l_z<l_v-l_x$ and $l_z-l_{\eta(x)}<l_v-l_x$. By the induction hypothesis, we now have $|S'_2|\leq |S(u,v,z)|$ and $|S'_1|\leq |S(x,z,\eta(x))|$. Therefore, $|S'|=1+|S'_1|+|S'_2|\leq 1+|S(x,z,\eta(x))|+|S(u,v,z)|$. Recalling that $l_x<l_z$, $L_z<r_x$ and that both $z$ and $x$ belong to a strong independent set $S'$ of $G$, we can conclude that $R_x<l_z$. This means that $z\in B$ and from equation~(\ref{eqn:max}), we now have $|S(u,v,x)|\geq |\{x\}\cup S(x,z,\eta(x))\cup S(u,v,z)|=1+|S(x,z,\eta(x))|+|S(u,v,z)|$ (as the sets $\{x\}$, $S(x,z,\eta(x))$ and $S(u,v,z)$ are pairwise disjoint). This shows that $|S(u,v,x)|\geq |S'|$.

Next, we shall consider the case when there does not exist any vertex $z\in S'\setminus\{x\}$ such that $L_z<r_x$. Then for every $w\in S'\setminus \{x\}$, we have $r_x<L_w$, which implies that $S'\setminus\{x\}\subseteq X(x,v)=Y(x,v,\eta(x))$. As $S'\setminus\{x\}$ is a strong independent set in $G$ and $l_v-l_{\eta(x)}<l_v-l_x$, we have $|S'\setminus\{x\}|\leq |S(x,v,\eta(x))|$ by our induction hypothesis.
Therefore, $|S'|=1+|S'\setminus\{x\}|\leq 1+|S(x,v,\eta(x))|=|\{x\}\cup S(x,v,\eta(x))|$ (note that $x\notin S(x,v,\eta(x))$). From equation~(\ref{eqn:max}), it is clear that $|S(u,v,x)|\geq |\{x\}\cup S(x,v,\eta(x))|$. We thus have $|S'|\leq |S(u,v,x)|$ as required.
\end{proof}
\begin{theorem}\label{thm:intnestmaxsi}
There is an $O(n^4)$ algorithm that computes a maximum cardinality strong independent set in an interval nest digraph $G$, given the interval nest representation of $G$ as input.
\end{theorem}
\begin{proof}
Add the intervals corresponding to two dummy vertices $a$ and $b$ to the input interval nest representation. Recalling that the leftmost endpoint in the input representation was 1 and the rightmost $4|V(G)|$, let $L_a=-4$, $l_a=-3$, $r_a=-1$, $R_a=0$, $L_b=-2$, $l_b=4|V(G)|+1$, $r_b=4|V(G)|+2$ and $R_b=4|V(G)|+3$. The sets $X(u,v)$ for all $u,v\in V(G)\cup\{a,b\}$ and $Y(u,v,x)$ for all $(u,v,x)\in (V(G)\cup\{a,b\})^3$ can be computed in $O(n^4)$ time. The algorithm then calls the procedure ComputeS($a,b,\eta(a)$) and outputs the set $S(a,b,\eta(a))$. Note that this being a dynamic programming algorithm, a call to $S(u,v,x)$ for some $(u,v,x)\in (V(G)\cup\{a,b\})^3$ is made only if $S(u,v,x)$ has not been computed before---or in other words, the algorithm ensures that a call to ComputeS($u,v,x$) is made at most once for each triple $(u,v,x)\in (V(G)\cup\{a,b\})^3$. Therefore, the total number of times the procedure ComputeS needs to be called recursively during the execution of ComputeS($a,b,\eta(a)$) is at most $(n+2)^3$. It is easy to see from the procedure ComputeS($u,v,x$) that the time spent in the computation of $S(u,v,x)$ outside the recursive calls to the procedure is $O(n)$. Therefore, the total running time of ComputeS($a,b,\eta(a)$) is $O(n^4)$, implying that our algorithm has time complexity $O(n^4)$. We only need to show that the output of the algorithm, $S(a,b,\eta(a))$, is a maximum cardinality strong independent set in $G$. It is clear that $X(a,b)=Y(a,b,\eta(a))=V(G)$. Therefore, by Lemmas~\ref{lem:issis} and~\ref{lem:ismax}, $S(a,b,\eta(a))$ is a maximum cardinality strong independent set in $G$.
\end{proof}

\subsection{The Uniquely Restricted Matching Problem in Interval Graphs}
Let $G$ be an interval graph for which we wish to compute a maximum cardinality uniquely restricted matching. Note that we can assume that the interval representation of the input graph $G$ is at our disposal. This is because even if the input graph is provided as an adjacency list, there are well-known algorithms that can generate an interval representation of $G$ in linear-time~\cite{BoothLueker,Corneiletal,Habibetal}. Let $\{I_u\}_{u\in V(G)}$ be an interval representation of $G$. For a vertex $u\in V(G)$, let $I_u=[l_u,r_u]$.

We shall define an interval nest digraph $D$ with $V(D)=E(G)$. The arcs of $D$ are defined by specifying the interval nest representation $\{(S_e,T_e)\}_{e\in V(D)}$ of $D$ as follows. For each $e=uv\in V(D)$, where $u,v\in V(G)$, we define $S_e = I_u\cup I_v$ and $T_e = I_u\cap I_v$. Clearly, for each $e\in V(D)$, we have $T_e\subseteq S_e$ and therefore this is an interval nest representation (note that the union or intersection of any two intervals that have a nonempty intersection is again an interval). Thus, $D$ is an interval nest digraph.

\begin{theorem}\label{thm:reduc}
Let $G$ and $D$ be as defined above.
Let $S\subseteq E(G)$. Then $S$ is a strong independent set in $D$ if and only if $S$ is a uniquely restricted matching in $G$.
\end{theorem}
\begin{proof}
Suppose that $S$ is a strong independent set in $D$. Let $e,e'\in S$ and let $e=uv$ and $e'=u'v'$. We first show that $e$ and $e'$ cannot be incident on a common vertex. Suppose for the sake of contradiction that the edges $e$ and $e'$ of $G$ share a common vertex. We shall assume without loss of generality that $v=v'$. Then clearly, $T_{e'}=(I_{u'}\cap I_{v'})\subseteq I_{v'}=I_v\subseteq (I_u\cup I_v)=S_e$, implying that $S_e\cap T_{e'}\neq\emptyset$ and therefore, $(e,e')\in E(D)$. Similarly, we have $T_e=(I_u\cap I_v)\subseteq I_v=I_{v'}\subseteq (I_{u'}\cup I_{v'})=S_{e'}$, leading us to infer that $(e',e)\in E(D)$. But this contradicts the fact that both $e$ and $e'$ belong to a strong independent set $S$ in $D$. Thus, we can conclude that the edges $e$ and $e'$ in $G$ have no common vertex, or in other words, $\{e,e'\}$ is a matching in $G$. Next, we show that there is no alternating cycle with respect to $\{e,e'\}$ in $G$. Suppose for the sake of contradiction that there is such a cycle. Then by Lemma~\ref{lem:altc4}, we know that in $G$, each of $u,v$ has at least one neighbour in $\{u',v'\}$ and each of $u',v'$ has at least one neighbour in $\{u,v\}$. This means that each of $I_u$ and $I_v$ intersects $I_{u'}\cup I_{v'}$ and each of $I_{u'}$ and $I_{v'}$ intersects $I_u\cup I_v$. Since $uv,u'v'\in E(G)$, this implies that $I_u\cap I_v$ intersects $I_{u'}\cup I_{v'}$ and $I_{u'}\cap I_{v'}$ intersects $I_u\cup I_v$. We thus have $T_e\cap S_{e'}\neq\emptyset$ and $T_{e'}\cap S_e\neq\emptyset$. By definition of $D$, it must then be the case that $(e,e'),(e',e)\in E(D)$. But this contradicts the fact that both $e$ and $e'$ belong to a strong independent set $S$ in $D$. We thus conclude that for any two edges $e,e'\in S$, $\{e,e'\}$ is a matching in $G$ and that there is no alternating cycle with respect to $\{e,e'\}$ in $G$. As $G$ is an interval graph, this implies, by Theorem~\ref{thm:urmintc4}, that $S$ is a uniquely restricted matching in $G$.

Now suppose that $S$ is a uniquely restricted matching in $G$. Again let $e,e'\in S$ and let $e=uv$ and $e'=u'v'$.
Suppose for the sake of contradiction that $(e,e'),(e',e)\in E(D)$. As $(e,e')\in E(D)$, we can infer that $S_e\cap T_{e'}\neq\emptyset$, which means that $I_{u'}\cap I_{v'}$ intersects $I_u\cup I_v$. Therefore, both $I_{u'}$ and $I_{v'}$ intersect at least one of $I_u$ or $I_v$. We can thus conclude that each of $u',v'$ is adjacent to at least one vertex in $\{u,v\}$. Now since $(e',e)\in E(D)$, we can follow the same arguments to reach the conclusion that each of $u,v$ is adjacent to at least one vertex in $\{u',v'\}$. From Lemma~\ref{lem:altc4}, we now have that there is an alternating cycle with respect to $\{e,e'\}$ in $G$. But this contradicts the fact that both $e$ and $e'$ belongs to a uniquely restricted matching $S$ in $G$. Therefore, for any pair of edges $e,e'\in S$, we have either $(e,e')\notin E(D)$ or $(e',e)\notin E(D)$, which allows us to conclude that $S$ is a strong independent set in $D$.
\end{proof}

\begin{theorem}
There is a polynomial-time algorithm that computes a maximum cardinality uniquely restricted matching in an interval graph.
\end{theorem}
\begin{proof}
We can generate an interval representation of the input graph $G$ in $O(n+m)$ time using any of the several well-known algorithms~(for example, \cite{Kratschetal}). The interval nest representation of the digraph $D$ corresponding to the interval representation of $G$ can be computed in $O(m)$ time. The algorithm described in the proof of Theorem~\ref{thm:intnestmaxsi} can now be used to compute a maximum cardinality strong independent set in $D$ in $O(m^4)$ time. It follows from Theorem~\ref{thm:reduc} that this strong independent set corresponds to a maximum cardinality uniquely restricted matching in $G$.
\end{proof}
\section{Concluding Remarks}
The complexity status of the problem of computing a maximum cardinality uniquely restricted matching in a permutation graph remains open. We note here that in permutation graphs, unlike interval graphs or bipartite permutation graphs, the fact that a matching $M$ is not uniquely restricted does not necessarily mean that there is an alternating cycle of length 4 with respect to $M$ in the graph. In fact, there does not exist any constant $k$ such that in any permutation graph with a matching $M$ that is not uniquely restricted, there exists an alternating cycle of length $k$ with respect to $M$. The following theorem states this fact.
\begin{theorem}\label{thm:permnoaltCk}
For every even integer $k\geq 4$, there exists a permutation graph $G$ with a matching $M$ in it such that the only alternating cycle with respect to $M$ in $G$ has length $k$.
\end{theorem}
\begin{proof}
For $k=4$, the $C_4$ and any perfect matching in it can be easily seen to satisfy the statement of the theorem. For $k\geq 6$, we construct the graph $G$ as shown in Figure~\ref{fig:perm}. The figure on the top left shows how to choose the matching $M$ when $k$ is not a multiple of 4 and the figure on the top right shows how to choose the matching $M$ when $k$ is a multiple of 4. It is clear that the only alternating cycle with respect to $M$ in $G$ in both cases is the Hamiltonian cycle in $G$, which has length $k$. It only remains to be shown that this graph is a permutation graph. We will use the fact that permutation graphs are exactly the graphs which have an intersection representation using line segments whose endpoints lie on two parallel lines~\cite{Golumbic}. It is easy to verify that the diagram at the bottom of Figure~\ref{fig:perm} shows such a representation of $G$ (in the figure, the endpoints of each line segment have been labelled with the name of the vertex that is represented by that line segment). The graph $G$ is therefore a permutation graph.

\end{proof}
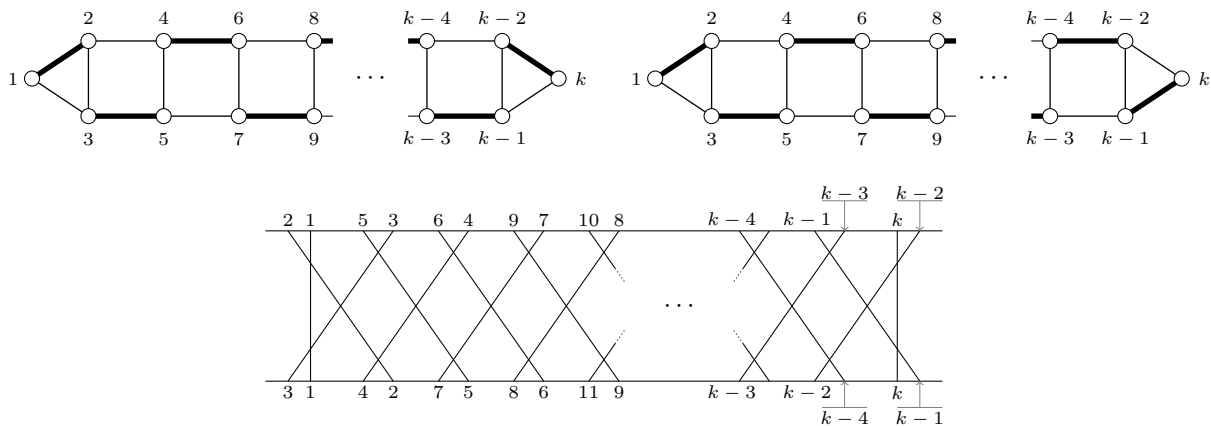
\begin{figure}
\begin{center}
\renewcommand{\vertexset}{(1,0.25,0.5),(2,1,1),(3,1,0),(4,2,1),(5,2,0),(6,3,1),(7,3,0),(8,4,1),(9,4,0),(k-4,5.5,1),(k-3,5.5,0),(k-2,6.5,1),(k-1,6.5,0),(k,7.25,0.5)}
\renewcommand{\edgeset}{(1,2,,2),(1,3),(2,4),(4,6,,2),(6,8),(3,5,,2),(5,7),(7,9,,2),(2,3),(4,5),(6,7),(8,9),(k-4,k-3),(k-2,k-1),(k-4,k-2),(k-3,k-1,,2),(k-2,k,,2),(k-1,k)}
\begin{tikzpicture}
\draw (4,0) -- (4.25,0);
\draw [line width=2] (4,1) -- (4.25,1);
\draw (5.5,0) -- (5.25,0);
\draw [line width=2] (5.5,1) -- (5.25,1);
\drawgraph
\node at (0,0.5) {\scriptsize 1};
\node at (1,1.3) {\scriptsize 2};
\node at (2,1.3) {\scriptsize 4};
\node at (3,1.3) {\scriptsize 6};
\node at (4,1.3) {\scriptsize 8};
\node at (5.5,1.3) {\scriptsize $k-4$};
\node at (6.5,1.3) {\scriptsize $k-2$};
\node at (1,-0.3) {\scriptsize 3};
\node at (2,-0.3) {\scriptsize 5};
\node at (3,-0.3) {\scriptsize 7};
\node at (4,-0.3) {\scriptsize 9};
\node at (5.5,-0.3) {\scriptsize $k-3$};
\node at (6.5,-0.3) {\scriptsize $k-1$};
\node at (7.55,0.5) {\scriptsize $k$};
\node at (4.75,0.5) {$\cdots$};
\end{tikzpicture}
\renewcommand{\vertexset}{(1,0.25,0.5),(2,1,1),(3,1,0),(4,2,1),(5,2,0),(6,3,1),(7,3,0),(8,4,1),(9,4,0),(k-4,5.5,1),(k-3,5.5,0),(k-2,6.5,1),(k-1,6.5,0),(k,7.25,0.5)}
\renewcommand{\edgeset}{(1,2,,2),(1,3),(2,4),(4,6,,2),(6,8),(3,5,,2),(5,7),(7,9,,2),(2,3),(4,5),(6,7),(8,9),(k-4,k-3),(k-2,k-1),(k-4,k-2,,2),(k-3,k-1),(k-2,k),(k-1,k,,2)}
\begin{tikzpicture}
\draw (4,0) -- (4.25,0);
\draw [line width=2] (4,1) -- (4.25,1);
\draw [line width=2] (5.5,0) -- (5.25,0);
\draw (5.5,1) -- (5.25,1);
\drawgraph
\node at (0,0.5) {\scriptsize 1};
\node at (1,1.3) {\scriptsize 2};
\node at (2,1.3) {\scriptsize 4};
\node at (3,1.3) {\scriptsize 6};
\node at (4,1.3) {\scriptsize 8};
\node at (5.5,1.3) {\scriptsize $k-4$};
\node at (6.5,1.3) {\scriptsize $k-2$};
\node at (1,-0.3) {\scriptsize 3};
\node at (2,-0.3) {\scriptsize 5};
\node at (3,-0.3) {\scriptsize 7};
\node at (4,-0.3) {\scriptsize 9};
\node at (5.5,-0.3) {\scriptsize $k-3$};
\node at (6.5,-0.3) {\scriptsize $k-1$};
\node at (7.55,0.5) {\scriptsize $k$};
\node at (4.75,0.5) {$\cdots$};
\end{tikzpicture}\vspace{0.1in}\\
\begin{tikzpicture}
\draw (0,0) -- (9,0);
\draw (0,2) -- (9,2);
\draw (0.3,2) -- (1.7,0);
\draw (0.3,0) -- (1.7,2);
\draw (1.3,2) -- (2.7,0);
\draw (1.3,0) -- (2.7,2);
\draw (2.3,2) -- (3.7,0);
\draw (2.3,0) -- (3.7,2);
\draw (3.3,2) -- (4.7,0);
\draw (3.3,0) -- (4.7,2);
\draw (6.3,2) -- (7.7,0);
\draw (6.3,0) -- (7.7,2);
\draw (7.3,2) -- (8.7,0);
\draw (7.3,0) -- (8.7,2);
\draw (4.3,2) -- (4.65,1.5);
\draw [densely dotted] (4.65,1.5) -- (4.79,1.3);
\draw (4.3,0) -- (4.65,0.5);
\draw [densely dotted] (4.65,0.5) -- (4.79,0.7);
\draw (6.7,2) -- (6.35,1.5);
\draw [densely dotted] (6.35,1.5) -- (6.21,1.3);
\draw (6.7,0) -- (6.35,0.5);
\draw [densely dotted] (6.35,0.5) -- (6.21,0.7);
\draw (0.6,2) -- (0.6,0);
\draw (8.4,2) -- (8.4,0);
\node at (0.6,2.15) {\scriptsize 1};
\node at (0.6,-0.15) {\scriptsize 1};
\node at (0.3,2.15) {\scriptsize 2};
\node at (1.7,-0.15) {\scriptsize 2};
\node at (0.3,-0.15) {\scriptsize 3};
\node at (1.7,2.15) {\scriptsize 3};
\node at (1.3,2.15) {\scriptsize 5};
\node at (2.7,-0.15) {\scriptsize 5};
\node at (1.3,-0.15) {\scriptsize 4};
\node at (2.7,2.15) {\scriptsize 4};
\node at (2.3,2.15) {\scriptsize 6};
\node at (3.7,-0.15) {\scriptsize 6};
\node at (2.3,-0.15) {\scriptsize 7};
\node at (3.7,2.15) {\scriptsize 7};
\node at (3.3,2.15) {\scriptsize 9};
\node at (4.7,-0.15) {\scriptsize 9};
\node at (3.3,-0.15) {\scriptsize 8};
\node at (4.7,2.15) {\scriptsize 8};
\node at (4.3,2.15) {\scriptsize 10};
\node at (4.3,-0.15) {\scriptsize 11};
\node at (8.4,2.15) {\scriptsize $k$};
\node at (8.4,-0.15) {\scriptsize $k$};
\node at (7.2,2.15) {\scriptsize $k-1$};
\node at (8.7,-0.5) {\scriptsize $k-1$};
\node at (7.2,-0.15) {\scriptsize $k-2$};
\node at (8.7,2.5) {\scriptsize $k-2$};
\node at (7.7,2.5) {\scriptsize $k-3$};
\node at (7.7,-0.5) {\scriptsize $k-4$};
\node at (6.2,2.15) {\scriptsize $k-4$};
\node at (6.2,-0.15) {\scriptsize $k-3$};
\draw [help lines] (8.4,2.4) -- (9,2.4);
\draw [help lines,->] (8.7,2.4) -- (8.7,2);
\draw [help lines] (8.4,-0.35) -- (9,-0.35);
\draw [help lines,->] (8.7,-0.35) -- (8.7,0);
\draw [help lines] (7.4,2.4) -- (8,2.4);
\draw [help lines,->] (7.7,2.4) -- (7.7,2);
\draw [help lines] (7.4,-0.35) -- (8,-0.35);
\draw [help lines,->] (7.7,-0.35) -- (7.7,0);
\node at (5.5,1) {$\cdots$};

\end{tikzpicture}
\end{center}
\caption{Construction of the graph $G$ in the proof of Theorem~\ref{thm:permnoaltCk}.}\label{fig:perm}
\end{figure}

Note that the algorithm described in Theorem~\ref{thm:intnestmaxsi} requires an interval nest representation of the interval nest digraph as input as the recognition problem for interval nest digraphs is still not known to be polynomial-time. On the other hand, the recognition problem for interval digraphs is known to be polynomial-time solvable~\cite{Muller}. It would be interesting to explore whether a polynomial-time algorithm can be constructed for the maximum cardinality strong independent set problem in interval digraphs.
\vspace{.1in}

\noindent\textbf{Acknowledgements.} We are grateful to Joydeep Mukherjee for his patient explanation of a result of his, which served as an inspiration for the proof of Theorem~\ref{thm:intnestmaxsi}. The first and second authors were partially supported by the ISI project CSRU2. The first and third authors were partially supported by the DST-INSPIRE Faculty Award IFA12-ENG-21.
\bibliography{urm}
\end{document}